\documentclass[11pt]{article}
\usepackage[hidelinks]{hyperref}
\usepackage{tikz}
\usetikzlibrary{calc}
\usepackage{multicol}
\usepackage{mdframed}
\usepackage{color}
\usepackage{epsfig,subfigure}
\usepackage{kbordermatrix}
\usepackage{amssymb, amsmath}
\usepackage{color}
\usepackage{bm}
\usepackage{amsthm}
\usepackage{authblk}
\usepackage{cite}
\usetikzlibrary{shapes,decorations,arrows,calc,arrows.meta,fit,positioning}

\usepackage{graphicx}
\usepackage{color}

\usepackage{amssymb}
\usepackage{amsmath,amsfonts,amssymb}
\usepackage{verbatim}
\usepackage{stfloats}
\usepackage[bookmarks=false]{}

\newtheorem{theorem}{Theorem}
\newtheorem{corollary}{Corollary}

\newtheorem{lemma}{Lemma}

\newtheorem{remark}{Remark}

\newcommand\blfootnote[1]{%
  \begingroup
  \renewcommand\thefootnote{}\footnote{#1}%
  \addtocounter{footnote}{-1}%
  \endgroup
}
\allowdisplaybreaks

\begin{document}

\title{On the Capacity of Secure Distributed Batch Matrix Multiplication}
\author{Zhuqing Jia and Syed A. Jafar}
\affil{Center for Pervasive Communications and Computing (CPCC), UC Irvine\\
Email: \{zhuqingj, syed\}@uci.edu}
\date{}
\maketitle

\begin{abstract}  
The problem of secure distributed batch matrix multiplication (SDBMM) studies the communication efficiency of retrieving a sequence of desired  matrix products ${\bf AB}$ $=$  $({\bf A}_1{\bf B}_1,$ ${\bf A}_2{\bf B}_2,$  $\cdots,$  ${\bf A}_S{\bf B}_S)$ from $N$ distributed servers where the constituent  matrices ${\bf A}=({\bf A}_1, {\bf A}_2, \cdots, {\bf A}_S)$ and  ${\bf B}=({\bf B}_1, {\bf B}_2,\cdots,{\bf B}_S)$ are stored in $X$-secure coded form, i.e., any group of up to $X$ colluding servers learn nothing about ${\bf A, B}$.  It is assumed that ${\bf A}_s\in\mathbb{F}_q^{L\times K}, {\bf B}_s\in\mathbb{F}_q^{K\times M}, s\in\{1,2,\cdots, S\}$ are uniformly and independently distributed and $\mathbb{F}_q$ is a large finite field. The rate of an SDBMM scheme is defined as the ratio of  the number of bits of desired information that is retrieved, to the total number of bits downloaded on average. The supremum of achievable rates is called the capacity of SDBMM. In this work we explore the capacity of SDBMM, as well as several of its variants, e.g., where the user may already have either ${\bf A}$ or ${\bf B}$ available as side-information, and/or where the security constraint for either ${\bf A}$ or ${\bf B}$ may be relaxed. We obtain  converse bounds, as well as  achievable schemes for  various cases of SDBMM, depending on the $L, K, M, N, X$ parameters, and identify parameter regimes where these bounds match. {In particular, the capacity for securely computing a batch of outer products of two vectors is $(1-X/N)^+$,   for a batch of inner products of two (long) vectors the capacity approaches $(1-2X/N)^+$ as the length of the vectors approaches infinity, and in general for sufficiently large $K$ (e.g., $K>2\min(L,M)$), the capacity $C$ is bounded as $(1-2X/N)^+\leq C<(1-X/N)^+$.} A remarkable aspect of our upper bounds is a connection between SDBMM and a form of private information retrieval (PIR) problem, known as multi-message $X$-secure $T$-private information retrieval (MM-XSTPIR). Notable features of our achievable schemes include the use of cross-subspace alignment and a  transformation  argument that converts a scalar multiplication problem into a scalar addition problem, allowing a surprisingly efficient solution. 
\end{abstract}

\blfootnote{This work is supported in part by funding from NSF grants CCF-1617504, CCF-1907053, CNS-1731384, ONR grant N00014-21-1-2386, and ARO grant W911NF1910344.}

\newpage

\section{Introduction}
Distributed matrix multiplication is a key building block for a variety of applications that include collaborative filtering, object recognition, sensing and data fusion, cloud computing, augmented reality and machine learning. Coding techniques, such as MDS codes \cite{Lee_Suh_Ramchandran},  Polynomial codes \cite{Yu_Ali_Avestimehr}, MatDot and PolyDot codes \cite{Dutta_Fahim_Haddadpour,GPdot}, and {Entangled Polynomial codes \cite{Yu_Maddah-Ali_Avestimehr}} have been shown to be capable of improving the efficiency of distributed matrix multiplication. However, with the expanding scope of distributed computing applications, there are mounting security concerns \cite{Lagrange,Yang_Lee, Chang_Tandon_SDMMOS, Kakar_Ebadifar_Sezgin_CSA,DOliveira_Rouayheb_Karpuk, Aliasgari_Simeone_Kliewer} about sharing information with external servers. The problem of secure distributed batch matrix multiplication (SDBMM) is motivated by these security concerns. 

\begin{figure}[!h]
\begin{center}
\begin{tikzpicture}[scale=0.75]
\node[circle, help lines, fill=black!5, text=black!50!white, draw=black, minimum size=0.7cm, inner sep=0] (S1) at (3cm, 1cm) { ${\bf M}_1$};
\node[circle, draw=black, fill=black!5, text=black,minimum size=0.7cm,  inner sep=0] (Si) at (6cm, 1cm) { ${\bf M}_i$};
\node[circle,  minimum size=0.7cm,  inner sep=0] (dots) at (4.5cm, 1cm) { $\cdots$};
\node[circle,  minimum size=0.7cm,  inner sep=0] (Mj) at (9cm, 1cm) { $\cdots$};
\node [draw, rectangle,fill=black!5, text=black, inner sep =0.2cm] (D1) at (2cm, -1.5cm) {\footnotesize Server $1$};
\node [rectangle, inner sep =0.2cm] (Ddots1) at (4cm, -1.5cm) {$\cdots$};
\node [draw, rectangle, fill=black!5, text=black, inner sep =0.2cm] (Dn) at (6cm, -1.5cm) {\footnotesize Server $n$};
\node [rectangle, inner sep =0.2cm] (Ddots2) at (8cm, -1.5cm) {$\cdots$};
\node [draw, rectangle, fill=black!5, text=black, inner sep =0.2cm] (DN) at (10cm, -1.5cm) {\footnotesize Server $N$};
\draw [help lines, ->] (S1)--(D1);
\draw [help lines, ->] (S1)--(Dn);
\draw [help lines, ->] (S1)--(DN);
\draw [ ->] (Si)--(D1) node[draw, rectangle, fill=white, pos=0.3]{\scriptsize $\widetilde{M}_i^1$};
\draw [ ->] (Si)--(Dn) node[draw, rectangle, fill=white, pos=0.3]{\scriptsize $\widetilde{M}_i^n$};
\draw [ ->] (Si)--(DN) node[draw, rectangle, fill=white, pos=0.3]{\scriptsize $\widetilde{M}_i^N$};
\node[circle, help lines, draw=black, fill=black!5, text=black, minimum size=0.7cm, inner sep=0] (U1) at (3cm, -4cm) { $U_1$};
\node[circle, help lines,  minimum size=0.7cm, inner sep=0] (Udots1) at (4.5cm, -4cm) { $\cdots$};
\node[circle, draw=black, fill=black!5, text=black, minimum size=0.7cm, inner sep=0] (Uj) at (6cm, -4cm) { $U_j$};
\node[circle, help lines,  minimum size=0.7cm, inner sep=0] (Udots2) at (7.5cm, -4cm) { $\cdots$};
\draw [help lines, ->] (D1)--(U1) node[ rectangle, pos=0.3]{};
\draw [help lines, ->] (Dn)--(U1) node[ rectangle,  pos=0.3]{};
\draw [help lines, ->] (DN)--(U1) node[ rectangle, pos=0.3]{};
\draw [thick, ->] (D1)--(Uj) node[draw, rectangle, fill=white, pos=0.3]{\scriptsize $\Delta_1$};
\draw [thick, ->] (Dn)--(Uj) node[draw, rectangle, fill=white, pos=0.3]{\scriptsize $\Delta_ n$};
\draw [thick, ->] (DN)--(Uj) node[draw, rectangle, fill=white, pos=0.3]{\scriptsize $\Delta_ N$};
\node[below=0.5cm of Uj, minimum size=0.3cm, inner sep=0.1cm] (P) {\scriptsize  ${\bf M}_k{\bf M}_l$};
\draw[->](Uj)--(P);

\begin{scope}[shift={(11,0)}]
\node[circle, help lines, fill=black!5, text=black, draw=black, minimum size=0.7cm, inner sep=0] (S1) at (3cm, 1cm) { ${\bf A}$};

\node[circle, draw=black, fill=black!5, text=black,minimum size=0.7cm,  inner sep=0] (Si) at (9cm, 1cm) { ${\bf B}$};

\node [draw, rectangle,fill=black!5, text=black, inner sep =0.2cm] (D1) at (2cm, -1.5cm) {\footnotesize Server $1$};
\node [rectangle, inner sep =0.2cm] (Ddots1) at (4cm, -1.5cm) {$\cdots$};
\node [draw, rectangle, fill=black!5, text=black, inner sep =0.2cm] (Dn) at (6cm, -1.5cm) {\footnotesize Server $n$};
\node [rectangle, inner sep =0.2cm] (Ddots2) at (8cm, -1.5cm) {$\cdots$};
\node [draw, rectangle, fill=black!5, text=black, inner sep =0.2cm] (DN) at (10cm, -1.5cm) {\footnotesize Server $N$};
\draw [help lines, ->] (S1)--(D1);
\draw [help lines, ->] (S1)--(Dn);
\draw [help lines, ->] (S1)--(DN);
\draw [  ->] (S1)--(DN) node[draw, rectangle, fill=white, pos=0.3]{\scriptsize $\widetilde{A}_N$};
\draw [  ->] (S1)--(D1) node[draw, rectangle, fill=white, pos=0.3]{\scriptsize $\widetilde{A}_1$};
\draw [  ->] (S1)--(Dn) node[draw, rectangle, fill=white, pos=0.3]{\scriptsize $\widetilde{A}_n$};

\draw [ ->] (Si)--(D1) node[draw, rectangle, fill=white, pos=0.3]{\scriptsize $\widetilde{B}_1$};
\draw [ ->] (Si)--(Dn) node[draw, rectangle, fill=white, pos=0.3]{\scriptsize $\widetilde{B}_n$};
\draw [ ->] (Si)--(DN) node[draw, rectangle, fill=white, pos=0.3]{\scriptsize $\widetilde{B}_N$};

\node[circle, draw=black, fill=black!5, text=black, minimum size=0.9cm, inner sep=0] (Uj) at (6cm, -4cm) { \small User};

\draw [thick, ->] (D1)--(Uj) node[draw, rectangle, fill=white, pos=0.3]{\scriptsize $\Delta_1$};
\draw [thick, ->] (Dn)--(Uj) node[draw, rectangle, fill=white, pos=0.3]{\scriptsize $\Delta_ n$};
\draw [thick, ->] (DN)--(Uj) node[draw, rectangle, fill=white, pos=0.3]{\scriptsize $\Delta_ N$};
\node[below=0.5cm of Uj, minimum size=0.3cm, inner sep=0.1cm] (P) {\scriptsize  ${\bf A}{\bf B}$};
\draw[->](Uj)--(P);
\end{scope}
\end{tikzpicture}
\end{center}
\caption{\it\small (Left) General context for SDBMM showing various sources that produce large amounts of data represented as matrices ${\bf M}_1, {\bf M}_2, \cdots$, and store it at $N$ distributed  servers in $X$-secure form, coded independently as $\widetilde{M}_i^n$. Various authorized users access these servers and retrieve products of their desired matrices based on the downloads that they request from all $N$ servers. Unlike PIR (private information retrieval) \cite{PIRfirst} problems there are no privacy constraints in SDBMM, so users can publicly announce which matrix products they wish to retrieve. (Right) The SDBMM problem considered in this work, where the goal is to minimize the average size of the total download for a generic user whose desired matrices are labeled ${\bf A}, {\bf B}$. }
\end{figure}
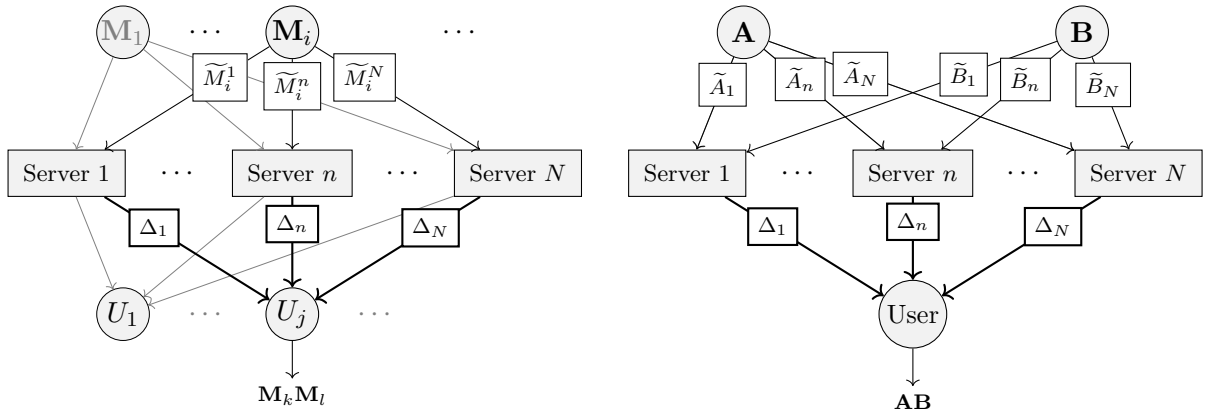

As defined in this work, SDBMM studies the communication efficiency of retrieving  desired matrix products from distributed servers where the constituent matrices are securely stored. Specifically, suppose ${\bf A}=({\bf A}_1, {\bf A}_2,\cdots, {\bf A}_S)$ and ${\bf B}=({\bf B}_1, {\bf B}_2, \cdots, {\bf B}_S)$ are collections of  random matrices that are stored across $N$ servers subject to an $X$-security guarantee, i.e., any colluding group of up to $X$ servers can learn nothing about the ${\bf A}$ and ${\bf B}$ matrices. Specifically, ${\bf A}_s\in\mathbb{F}_q^{L\times K}$, ${\bf B}_s\in\mathbb{F}_q^{K\times M}$  for all $s\in\{1,2,\cdots, S\}$, are independent and uniformly distributed, and $\mathbb{F}_q$ is assumed to be a large finite field. 
A user wishes to retrieve ${\bf A}{\bf B}=({\bf A}_1{\bf B}_1, {\bf A}_2{\bf B}_2, \cdots, {\bf A}_S{\bf B}_S)$, where each ${\bf A}_s{\bf B}_s, s\in\{1,2,\cdots, S\}$ is an $L\times M$ matrix in $\mathbb{F}_q$,  while downloading as little information from the $N$ servers as possible. The rate of an SDBMM scheme is defined as the ratio $H({\bf AB})/D$, where  $H({\bf AB})$ is the number of bits of desired information that is retrieved, and $D$ is the total number of bits downloaded on average. The supremum of achievable rates is called the capacity of SDBMM. In this work we study the capacity of SDBMM, as well as several of its variants, e.g., where the user may already have either ${\bf A}$ or ${\bf B}$ available as side-information,\footnote{The definition of rate takes into account the side-information available at the user. For example, say, the user has ${\bf B}$ available as side-information, then the rate is defined as $H({\bf AB}\mid {\bf B})/D$.} and/or where the security constraint for either ${\bf A}$ or ${\bf B}$ may be relaxed. 

{SDBMM may be viewed as a generalization of the problem of secure distributed matrix multiplication (SDMM) where the batch size is set to  $S=1$. Information-theoretic study of SDMM was introduced recently in \cite{Chang_Tandon_SDMMOS}, under a closely related model with a few subtle differences. Motivated by a master-worker model of distributed computation, it is assumed in \cite{Chang_Tandon_SDMMOS} that the ${\bf A}, {\bf B}$ matrices originate at the user (master), who securely encodes and sends these matrices to $N$ servers (workers), and then from just the downloads that he receives from the servers in return, the user is able to compute ${\bf AB}$. The goal in \cite{Chang_Tandon_SDMMOS}, as in this work, is to maximize the rate, defined as $H({\bf AB})/D$. The same model is  pursued in \cite{Kakar_Ebadifar_Sezgin_CSA, DOliveira_Rouayheb_Karpuk, Aliasgari_Simeone_Kliewer}, where new coding schemes are proposed for SDMM. However, certain aspects of the model have the potential to appear inconsonant. For example, a key assumption in  these models is that the user must not use his prior knowledge of ${\bf A}, {\bf B}$ and must decode ${\bf AB}$ only from the downloads. Since the goal is to minimize the communication cost, and the user already knows ${\bf A}, {\bf B}$, why not do the computation locally and avoid all communication entirely? Indeed, the question is not merely philosophical, because as shown in this work, in some cases, the best scheme even with the model of \cite{Chang_Tandon_SDMMOS, Kakar_Ebadifar_Sezgin_CSA,DOliveira_Rouayheb_Karpuk, Aliasgari_Simeone_Kliewer}  turns out to be one that allows the user to retrieve both ${\bf A}, {\bf B}$ from the downloads, and then compute ${\bf AB}$ locally --- something that could be done without the need for any communication (albeit it may incur a greater local computation cost on the part of the master, an important concern that is not reflected in the communication cost metric) if indeed ${\bf A}, {\bf B}$ originated at the user. On the other hand, the SDBMM model assumed in this work assumes that ${\bf A},{\bf B}$ do not both originate at the user,\footnote{One of ${\bf A}, {\bf B}$ may be available to the user as side-information in some variants of SDBMM studied in this work.} rather they are stored securely and remotely at the $N$ servers,\footnote{We envision that the matrices ${\bf A}$ and ${\bf B}$ comprise sensitive data that originates at other sources, and is securely stored at the $N$ servers. The communication cost of uploading the data to the servers is not a focus of this work, because the upload requires communication between a different set of entities (sources and servers) with their own separate communication channels and cost dynamics. Furthermore, if the data is of interest to many (authorized) users who perform their desired computations, then the repeated cost of such downloads could very well  outweigh the one-time cost of uploading the secured data to the servers.} thus eliminating this concern. Another difference between our model of SDMM and previous works is that the precise assumptions regarding matrix dimensions $L, K, M$ are left unclear in \cite{Chang_Tandon_SDMMOS, Kakar_Ebadifar_Sezgin_CSA}, indeed these dimensions do not appear in the capacity results in \cite{Chang_Tandon_SDMMOS, Kakar_Ebadifar_Sezgin_CSA}.  However, our model allows $L, K, M$ to take arbitrary values, including large values. As it turns out, our results reveal that the relative size of $L,K,M$ does matter for the capacity of SDBMM.

Using prior works on SDMM that are most closely related to this work, in particular \cite{Chang_Tandon_SDMMOS, Kakar_Ebadifar_Sezgin_CSA}, as the baseline, the main contributions of this work may be summarized as follows. In terms of converse (upper) bounds, in addition to the converse bounds that follow directly from \cite{Chang_Tandon_SDMMOS} (e.g., the converse for Theorem \ref{thm:cba} is inherited from \cite{Chang_Tandon_SDMMOS}), we find new upper bounds on SDBMM capacity by identifying a connection between SDBMM and a form of private information retrieval (PIR) problem, known as multi-message $X$-secure $T$-private information retrieval (MM-XSTPIR). By developing converse bounds for MM-XSTPIR\footnote{Notably while the capacity of multi-message PIR has been explored in \cite{Banawan_Ulukus_Multimessage,Li_Gastpar} and that of $X$-secure $T$-private PIR has been explored in \cite{Jia_Sun_Jafar_XSTPIR},  to our knowledge there has been no prior work on MM-XSTPIR.} and by using the connection of MM-XSTPIR to SDBMM, we are able to bound the capacity of SDBMM. In terms of lower bounds (achievability), a notable aspect of our achievable schemes is the idea of cross-subspace alignment (CSA) which was introduced in \cite{Jia_Sun_Jafar_XSTPIR}, and was recently applied to SDBMM in \cite{Kakar_Ebadifar_Sezgin_CSA} based on matrix partitioning. In this work, instead of coding across matrix partitions, we utilize CSA codes for batch processing (subsequent generalizations of this idea appear in   \cite{Jia_Jafar_MDSXSTPIR, Jia_Jafar_CDBC}). Another interesting aspect of the achievable schemes in this work is a transformation  that converts a scalar multiplication problem into a scalar addition problem. This transformation allows a surprisingly\footnote{The achieved rate  exceeds an upper bound previously obtained in literature \cite{Kakar_Ebadifar_Sezgin_CSA}. See the discussion following Theorem \ref{thm:cabphi}.} efficient (and capacity optimal) solution for scalar multiplication, outer products of vectors, and Hadamard products of matrices. Depending on the $L, K, M, N, X$ parameters, we identify several  regimes where the converse (upper) and achievability (lower) bounds match, settling the capacity of SDMM for those cases.  In contrast with \cite{Chang_Tandon_SDMMOS, Kakar_Ebadifar_Sezgin_CSA} where the matrix dimensions $L, K, M$ do not appear in the capacity results, our results reveal that the relative size of $L,K,M$ does matter. For example, capacity-achieving schemes for SDMM from \cite{Chang_Tandon_SDMMOS, Kakar_Ebadifar_Sezgin_CSA} fall short if $K/L<1$, and a converse bound for two-sided security that is derived in \cite{Kakar_Ebadifar_Sezgin_CSA} is violated under certain conditions as well (see discussion following Theorem \ref{thm:cabphi} and Theorem \ref{thm:cba} in this work). Indeed, we find that in general the capacity depends on all matrix dimensions as well as the parameters $N, X$. For example,  the capacity of a batch of outer products of two vectors is $(1-X/N)^+$,  the capacity of a batch of inner products of two (long) vectors approaches $(1-2X/N)^+$ as the length of the vectors approaches infinity, and in general for sufficiently large $K$ (e.g., $K>2\min(L,M)$), the capacity $C$ is bounded as $(1-2X/N)^+\leq C<(1-X/N)^+$ (see Theorem \ref{thm:cabphi} and Remark \ref{rem:highlights} for details). Finally, it may also be of interest to  consider a notion of dimension-independent capacity, defined as the infimum of SDBMM capacity over all possible matrix dimensions. Remarkably, the dimension-independent capacities for all cases studied in this work are settled as immediate corollaries of our main results.

}

{\it Notations:} {The notation $x^+$ denotes $\max(x,0)$.} For a positive integer $N$, $[N]$ stands for the set $\{1,2,\dots,N\}$. The notation $X_{[N]}$ denotes the set $\{X_1,X_2,\dots,X_N\}$. {For an index set $\mathcal{I}=\{i_1,i_2,\dots,i_n\}$, where $(i_1,i_2,\dots,i_n)$ are sorted in ascending order, let $X_{\mathcal{I}}$ denotes the tuple $(X_{i_1},X_{i_2},\dots,X_{i_n})$. If $A$ is a set/tuple of random variables, then by $H_q(A)$ we denote the joint entropy (base $q$) of those random variables. Mutual informations between sets of random variables are similarly defined as $I_q(\cdot)$.} We use the notation $X\sim Y$ to indicate that $X$ and $Y$ are identically distributed.

\section{Problem Statement}
\subsection{Problem Statement: SDBMM}
Let $\mathbf{A}=(\mathbf{A}_1, \mathbf{A}_2,\cdots,\mathbf{A}_S)$ represent $S$ random matrices, chosen independently and uniformly from all matrices over $\mathbb{F}_q^{L\times K}$. Similarly, let $\mathbf{B}=(\mathbf{B}_1, \mathbf{B}_2,\cdots,\mathbf{B}_S) $ represent $S$ random matrices, chosen independently and uniformly from all matrices over $\mathbb{F}_q^{K\times M}$. The user wishes to compute the products $\mathbf{AB}=(\mathbf{A}_1\mathbf{B}_1, \mathbf{A}_2\mathbf{B}_2, \cdots, \mathbf{A}_S\mathbf{B}_S)$.

The independence between matrices $\mathbf{A}_{[S]}$, $\mathbf{B}_{[S]}$ is formalized as follows.
\begin{equation}
H_q(\mathbf{A},\mathbf{B})=H_q(\mathbf{A})+H_q(\mathbf{B})=\sum_{s\in[S]}H_q(\mathbf{A}_s)+\sum_{s\in[S]}H_q(\mathbf{B}_s)\label{eq:matind}.
\end{equation}
Since we are operating over $\mathbb{F}_q$, let us express all entropies in base $q$ units.

The ${\bf A},{\bf B}$ matrices are available at $N$ servers with security levels $X_A, X_B$, respectively. This means that any group of up to $X_A$ colluding servers can learn nothing about ${\bf A}$ matrices, and any group of up to $X_B$ colluding servers can learn nothing about the ${\bf B}$ matrices.\footnote{For the most part, we will focus on cases with $X_A,X_B\in\{0,X\}$, i.e., the security level can either be $X>0$ or zero, where a security level zero implies that there is no security constraint for that set of matrices.} 
Security is achieved by  coding according to  secret sharing schemes that separately generate shares $\widetilde{A}_s^n, \widetilde{B}_s^n$ corresponding to each ${\bf A}_s, {\bf B}_s$, and make these shares\footnote{If $X_A=0$ then  we could choose $\widetilde{A}_s^n={\bf A}_s$. Similarly, if $X_B=0$, then it is possible to have $\widetilde{B}_s^n={\bf B}_s$.} available to the $n^{th}$ server, for all $n\in[N]$. The independence between these securely coded matrices is formalized as,
\begin{align}
I_q({\bf A}, \widetilde{A}_{[S]}^{[N]}; {\bf B}, \widetilde{B}_{[S]}^{[N]})&=0,\label{eq:sepencode}\\
H_q\left(\widetilde{A}_{[S]}^{[N]},\widetilde{B}_{[S]}^{[N]}\right)&=\sum_{s\in[S]}H_q(\widetilde{A}_{s}^{[N]})+\sum_{s\in[S]}H_q(\widetilde{B}_{s}^{[N]}).
\end{align}
Each matrix  must be recoverable from all its secret shares,
\begin{align}\label{eq:emfunc}
H_q(\mathbf{A}_s\mid \widetilde{A}_s^{[N]})=0, &&H_q(\mathbf{B}_s\mid \widetilde{B}_s^{[N]})=0, &&\forall s\in[S].
\end{align}
The ${\bf A}, {\bf B}$ matrices must be perfectly secure from any set of secret shares that can be accessed by a set of up to  $X_A, X_B$ colluding servers, respectively.
\begin{align}
I_q\left({\bf A}; \widetilde{A}_{[S]}^{\mathcal{X}}\right)&=0,&&\mathcal{X}\subset[N], &&|\mathcal{X}|=X_A\label{eq:securea}\\
I_q\left({\bf B}; \widetilde{B}_{[S]}^{\mathcal{X}}\right)&=0,&&\mathcal{X}\subset[N], &&|\mathcal{X}|=X_B \label{eq:secureb}\\
I_q({\bf A, B}; \widetilde{A}_{[S]}^{\mathcal{X}}, \widetilde{B}_{[S]}^{\mathcal{X}})&=0,&&\mathcal{X}\subset[N], &&|\mathcal{X}|=\min(X_A,X_B) \label{eq:secure}
\end{align}
In order to retrieve the products ${\bf A}{\bf B}$,  from each server $n\in[N]$, the user downloads $\Delta_n$ which is function of $\widetilde{A}_{[S]}^n, \widetilde{B}_{[S]}^n$.
\begin{equation}
H_q\left(\Delta_n\mid\widetilde{A}_{[S]}^n, \widetilde{B}_{[S]}^n\right)=0.
\end{equation}
The  side-information available to the user \emph{apriori} is denoted $\Psi$, which can be either the ${\bf A}$ matrices, or the ${\bf B}$ matrices, or null ($\phi$) if the user has no side-information. Given the downloads from all $N$ servers and the side-information, the user must be able to recover the matrix products ${\bf AB}$.
\begin{align}
H_q\left({\bf AB}\mid \Delta_{[N]}, \Psi\right)&=0.\label{eq:correct}
\end{align}
Let us define the rate of an SDBMM scheme as follows.
\begin{align}
R&=\frac{H_q({\bf AB}\mid \Psi)}{D}
\end{align}
where $D$ is the average value (over all realizations of ${\bf A}, {\bf B}$ matrices) of the total number of  $q$-ary symbols downloaded by the user from all $N$ servers. In order to steer away from the field-size concerns that are best left to coding-theoretic studies, we will only allow the field size to be asymptotically large, i.e., $q\rightarrow\infty$.
The capacity of SDBMM is the supremum of achievable rate values over all SDBMM schemes and over all $S$. 

{The results in this work show that the relative size of matrix dimensions $L,K,M$ does matter for SDBMM capacity results. However, it may  be of interest to also consider a notion of dimension-independent capacity, defined as the infimum of SDBMM capacity over all possible dimensions $L,K$ and $M$.}

\begin{remark} The goal of the SDBMM problem is to design schemes to minimize $D$. The normalization factor $H_q({\bf AB}\mid \Psi)$ is not particularly important since it does not depend on the scheme, it is simply a baseline that is chosen to represent the average download needed from a centralized server that directly sends ${\bf AB}$ to the user in the absence of security constraints. Other baselines, e.g., $H_q({\bf AB})$ may be chosen instead as in \cite{Chang_Tandon_SDMMOS}, or one could equivalently formulate the problem directly as a minimization of download cost $D$. We prefer the formulation as a rate maximization because it allows a more direct connection to the capacity of PIR, one of the main themes of this work. 
\end{remark}

Finally, depending upon which matrices are secured and/or available as side-information, we have the following versions of the SDBMM problem.
{\begin{align}
\begin{array}{|l|c|c|c|c|}\hline
\mbox{SDBMM version}&\mbox{secure}&\mbox{side-information $\Psi$}&\mbox{capacity}&\mbox{dimension-ind. capacity}\\\hline
&&&&\\[-1em]
\mbox{SDBMM}_{(\bf AB, \phi)}&{\bf A,B}& \phi & C_{(\bf AB, \phi)} & \mathring{C}_{(\bf AB, \phi)}\\\hline
&&&&\\[-1em]
\mbox{SDBMM}_{(\bf AB, B)}&{\bf A,B}& {\bf B} & C_{(\bf AB, B)}& \mathring{C}_{(\bf AB, B)}\\\hline
&&&&\\[-1em]
\mbox{SDBMM}_{(\bf B, \phi)}&{\bf B}& \phi & C_{(\bf B,\phi)}& \mathring{C}_{(\bf B,\phi)}\\\hline
&&&&\\[-1em]
\mbox{SDBMM}_{(\bf B, A)}&{\bf B}& {\bf A} & C_{(\bf B, \bf A)}& \mathring{C}_{(\bf B, \bf A)}\\\hline
&&&&\\[-1em]
\mbox{SDBMM}_{(\bf B, B)}&{\bf B}& {\bf B} & C_{(\bf B,\bf B)}&\mathring{C}_{(\bf B,\bf B)}\\\hline
\end{array}
\end{align}}
Thus, the version of SDBMM is indicated by the subscript which has two elements, the first representing the matrices that are secured and the second representing the matrices available to the user as side-information. Note that other cases, such as  $\mbox{SDBMM}_{(\bf AB, A)}$, $\mbox{SDBMM}_{(\bf A, \phi)}$, $\mbox{SDBMM}_{(\bf A, A)}$, $\mbox{SDBMM}_{(\bf A, B)}$, are equivalent to, $\mbox{SDBMM}_{(\bf AB, B)}$, $\mbox{SDBMM}_{(\bf B, \phi)}$, $\mbox{SDBMM}_{(\bf B, B)}$, $\mbox{SDBMM}_{(\bf B, A)}$,  respectively, by the inherent symmetry of the problem, leaving us with just the $5$ cases tabulated above.

\subsection{Problem Statement: Multi-Message $X$-Secure $T$-Private Information Retrieval}
{Since our converse bounds rely on a connection between SDBMM and multi-message $X$-secure $T$-private information retrieval (MM-XSTPIR), in this subsection, we formally define the problem of MM-XSTPIR.} Consider $K$ messages stored at $N$ distributed servers, $W_1, W_2,\dots, W_K$. Each message is represented by $L$ random symbols from the finite field $\mathbb{F}_q$.
\begin{align}
H_q(W_1)=H_q(W_2)=\dots=H_q(W_K)=L,\label{eq:mmxstpirmsgentropy}\\
H_q(W_{[K]})=KL,\label{eq:mmxstpirmsgind}
\end{align}
The information stored at the $n$-th server is denoted by $S_n$, $n\in[N]$. $X$-secure storage, $0\leq X\leq N$, guarantees that any $X$ (or fewer) colluding servers learns nothing about messages.
\begin{align}\label{eq:mmxstpirsecur}
I_q(S_{\mathcal{X}};W_{[K]})=0,\quad \forall\mathcal{X}\subset[N], |\mathcal{X}|=X.
\end{align}
To make information retrieval possible, messages must be function of $S_{[N]}$.
\begin{align}\label{eq:mmxstpirmsgfunc}
H_q(W_{[K]}|S_{[N]})=0.
\end{align}
The multi-message $T$-private information retrieval allows the user to retrieve $M$ messages simultaneously. The user privately and uniformly generate a set of indices of desired messages $\mathcal{K}$, $\mathcal{K}\subset[K], |\mathcal{K}|=M$. To retrieve desired messages privately, the user generates $N$ queries $Q_{[N]}^{\mathcal{K}}$. The $n$-th query $Q_{n}^{\mathcal{K}}$ is sent to the $n$-th server. The user has no prior knowledge of information stored at servers, i.e.,
\begin{align}\label{eq:mmxstpirqsind}
I_q(S_{[N]};\mathcal{K},Q_{[N]}^{\mathcal{K}})=0.
\end{align}
$T$-privacy, $0\leq T\leq N$, guarantees that any $T$ (or fewer) colluding servers learns nothing about $\mathcal{K}$.
\begin{align}\label{eq:mmxstpirpriv}
I_q(Q_{\mathcal{T}}^{\mathcal{K}},S_{\mathcal{T}};\mathcal{K})=0,\quad \forall\mathcal{T}\subset[N], |\mathcal{T}|=T.
\end{align}
Upon receiving user's query $Q_{n}^{\mathcal{K}}$, the $n$-th server responds user with an answer $A_n^{\mathcal{K}}$, which is function of the query and its storage, i.e.,
\begin{align}\label{eq:mmxstpiransfunc}
H_q(A_n^{\mathcal{K}}|Q_{n}^{\mathcal{K}},S_n)=0.
\end{align}
The user must be able to recover desired messages $W_{\mathcal{K}}$ from all answers $A_{[N]}^{\mathcal{K}}$.
\begin{align}\label{eq:mmxstpirdecode}
H_q(W_{\mathcal{K}}|A_{[N]}^{\mathcal{K}},Q_{[N]}^{\mathcal{K}},\mathcal{K})=0.
\end{align}
The rate of a multi-message XSTPIR scheme is defined by the number of $q$-ary symbols of desired messages that are retrieved per downloaded $q$-ary symbol,
\begin{align}
R=\frac{H_q(W_{\mathcal{K}})}{\sum_{n\in[N]}A_{n}^{\mathcal{K}}}=\frac{ML}{D}.
\end{align}
$D=\sum_{n\in[N]}A_{n}^{\mathcal{K}}$ is expected number of downloaded $q$-ary symbols from all servers. The capacity of multi-message XSTPIR is the supremum of rate over all feasible schemes, denoted as $C_{\text{MM-XSTPIR}}(N,X,T,K,M)$.

Note that setting $X=0$ reduces the problem to basic multi-message $T$-private information retrieval where storage is not secure. The setting $T=0$ reduces the problem to $X$-secure storage with no privacy constraint.

\section{Results}\label{sec:observ}
\subsection{A Connection between SDBMM and MM-XSTPIR}
{Let us begin by identifying a connection between SDBMM and MM-XSTPIR.}
\begin{lemma}\label{lemma:main}
The following bounds apply.
\begin{align}
K\geq M&&\implies&&\max\left(C_{(\bf AB,B)}, C_{(\bf B,B)},  C_{(\bf AB,\phi)}, C_{(\bf B,\phi)}\right)\leq C_{\text{MM-XSTPIR}}(N,X_A,X_B,K,M),\label{eq:lemmamain1}
\\
K\geq L&&\implies&&C_{(\bf AB,\phi)}\leq C_{\text{MM-XSTPIR}}(N,X_B,X_A,K,L),\label{eq:lemmamain2}
\end{align}
where $C_{\text{MM-XSTPIR}}(N,X,T,K,M)$ is the capacity of MM-XSTPIR with $N$ servers, $X$-secure storage and $T$-private queries, retrieving $M$ out of $K$ messages.
\end{lemma}
\begin{proof}
Let us first prove  the bound in  \eqref{eq:lemmamain1}, by showing that when $K\geq M$, then any SDBMM scheme where the side-information  available to the user is not\footnote{${\bf A}$ cannot be  the side-information because as noted, the transformation from SDBMM to MM-XSTPIR  interprets ${\bf A}$ as the data, which cannot be already available to the user in MM-XSTPIR. On the other hand, ${\bf B}$ may be included in the side-information because it is interpreted as the queries, which are automatically known to the user in MM-XSTPIR. Note that if ${\bf B}$ is also not included in the side-information, that just means that the user can retrieve ${\bf AB}$ from the downloads without using the knowledge of ${\bf B}$ in the resulting MM-XSTPIR scheme.}  ${\bf A}$, and where $X_B\neq 0$, automatically yields an MM-XSTPIR$(N,X_A, X_B, K, M)$ scheme with the same rate, essentially by thinking of ${\bf A}$ as the data and ${\bf B}$ as the query. Consider MM-XSTPIR with $K$ independent messages, each of which consists of $L$ i.i.d. uniform symbols in $\mathbb{F}_q$, say arranged in a column. For all $k\in[K]$, arrange these columns to form the matrix ${A}_1$ so that the $k^{th}$ column of ${ A}_1$ represents the $k^{th}$ message. Let the $M$ desired message indices be represented by the corresponding columns of the $K\times K$  identity matrix, and let these $M$ columns be arranged to form the $K\times M$ matrix ${ B}_1$.
Note that retrieving the matrix product ${ A}_1{ B}_1$ is identical to retrieving the $M$ desired messages. Now any $(X_A, X_B)$ secure SDBMM scheme with $S=1$ that does not have ${\bf A}$ as side-information, conditioned on the realizations ${\bf A}_1=A_1, {\bf B}_1=B_1$, yields an MM-XSTPIR scheme by treating $\widetilde{A}_1^n$ as the $X_A$-secure data stored at the $n^{th}$ server and $\widetilde{B}_1^n$ as the $X_B$-private query sent by the user to the $n^{th}$ server, for all $n\in[N]$. For arbitrary $S>1$ we can simply extend the data by a factor of $S$, i.e., each message is comprised of $SL$ symbols, so that the $k^{th}$ message is represented by the $k^{th}$ columns of ${ A}_1, { A}_2,\cdots, {A}_S$, treated as realizations of ${\bf  A}_1, {\bf  A}_2,\cdots, {\bf A}_S$. Thus, conditioning on the realization ${\bf B}_1={\bf B}_2=\cdots={\bf B}_S=B_1$ gives us the $S$-fold extension of the same scheme.  Furthermore, since ${\bf B}$ matrices are secure, i.e., $X_B>0$ for all SDBMM settings that appear in  \eqref{eq:lemmamain1}, it follows that ${\bf B}$ is independent of $\widetilde{B}_{[S]}^n$ for any $n\in[N]$. This in turn implies that ${\bf B}$ is independent of the download $\Delta_n$ received from Server $n$. Therefore, conditioning on  ${\bf B}$ taking values in the set that corresponds to MM-XSTPIR (i.e., $B_1$ restricted to any choice of  $M$ columns of the $K\times K$ identity matrix) does not affect the distribution of $\Delta_n$, or the entropy $H_q(\Delta_n)$. In other words, the average download of the SDBMM scheme remains unchanged  as it is specialized to yield an MM-XSTPIR scheme as described above. Now, since the number of desired $q$-ary symbols retrieved by this feasible MM-XSTPIR scheme is $SLM$, the  average download, $D$ for the SDBMM scheme cannot be less than $SLM/C_{\text{MM-XSTPIR}}(N,X_A,X_B,K,M)$. Therefore, we have a bound on the rate of the SDBMM scheme as
\begin{align}
R&= \frac{H_q({\bf AB}\mid {\bf B})}{D}\\
&\leq \frac{H_q({\bf AB})}{D}\label{eq:cond}\\
&\leq\frac{SLM}{D}\label{eq:SLM}\\
&\leq \frac{SLM}{SLM/C_{\text{MM-XSTPIR}}(N,X_A,X_B,K,M)}\\
&=C_{\text{MM-XSTPIR}}(N,X_A,X_B,K,M)
\end{align}
In \eqref{eq:cond} we used the fact that conditioning reduces entropy so that $H_q({\bf AB}\mid {\bf B})\leq H_q({\bf AB})$. In \eqref{eq:SLM} we used the fact that the matrix ${\bf AB}$ has $SLM$ elements from $\mathbb{F}_q$, and since the uniform distribution maximizes entropy, $H_q({\bf AB})\leq SLM$. {The bound in \eqref{eq:lemmamain2} is similarly shown, by treating $\widetilde{B}_1^n$ as the $X_B$-secure data stored at the $n^{th}$ server and $\widetilde{A}_1^n$ as the $X_A$-private query sent by the user to the $n^{th}$ server.}
This completes the proof of Lemma \ref{lemma:main}.
\end{proof}

\subsection{An Upperbound on the Capacity of MM-XSTPIR}
Motivated by Lemma \ref{lemma:main}, an upper bound on the capacity of MM-XSTPIR is  presented in the following theorem.

\begin{theorem}\label{thm:mmxstpir}
The capacity of MM-XSTPIR is bounded as follows.
\begin{align}
&C_{\text{MM-XSTPIR}}(N,X,T,K,M)\notag\\
\leq&\left\{
\begin{aligned}
&0,&&N\leq X,\\
&\frac{M(N-X)}{KN},&&X<N\leq X+T,\\
&\frac{N-X}{N}\left(1+\left(\frac{T}{N-X}\right)+\dots+\left(\frac{T}{N-X}\right)^{\lfloor\frac{K}{M}\rfloor-1}\right)^{-1},&&N>X+T.
\end{aligned}
\right.
\end{align}
\end{theorem}
The proof of Theorem \ref{thm:mmxstpir} is presented in Appendix \ref{sec:proofmmxstpir}. Note that Theorem \ref{thm:mmxstpir} also works for trivial security or privacy, i.e., when $X=0$ or $T=0$.

\begin{remark} In fact a closer connection exists between SDBMM and MM-XSTPC, i.e., multi-message $X$-secure $T$-private (linear) computation problem that is an extension of the private computation problem studied in \cite{Sun_Jafar_PC}. However, we use only  the connection to MM-XSTPIR because this setting is simpler and the  connection between SDBMM and MM-XSTPIR suffices for our purpose .
\end{remark}

\subsection{Entropies of Products of Random Matrices}
The following lemma is needed to evaluate the numerator in the rate expressions for SDBMM schemes.
\begin{lemma}\label{lemma:Hab}
Let $\mathbf{A}$, $\mathbf{B}$ be random matrices independently and uniformly distributed over $\mathbb{F}_q^{L\times K}$, $\mathbb{F}_q^{K\times M}$, respectively. As $q\rightarrow\infty$, we have
\begin{align}
H_q(\mathbf{A} \mathbf{B})&=\left\{
\begin{aligned}
&LM,&&K\geq \min(L,M)\\
&LK+KM-K^2,&&K<\min(L,M)
\end{aligned}
\right.,\label{eq:Hab}\\
H_q(\mathbf{A} \mathbf{B} \mid \mathbf{A})&=\min(LM,KM),\label{eq:Haba}\\
H_q(\mathbf{A} \mathbf{B} \mid \mathbf{B})&=\min(LM,LK),\label{eq:Habb}
\end{align}
in $q$-ary units.
\end{lemma}
The proof of Lemma \ref{lemma:Hab} appears in Appendix \ref{proof:lemmaHab}.

We now proceed to capacity characterizations for the various SDBMM models.
\subsection{Capacity of SDBMM$_{({\bf AB},\phi)}$}
Let us start with the basic SDBMM setting, where both matrices ${\bf A}, {\bf B}$ are $X$-secured, and there is no prior side-information available to the user. This is essentially the two-sided secure SDBMM setting considered previously in \cite{Chang_Tandon_SDMMOS, Kakar_Ebadifar_Sezgin_CSA}.
{\begin{theorem}\label{thm:cabphi}
The capacity of SDBMM$_{({\bf AB},\phi)}$, with $X_A=X_B=X$, is characterized under various settings as follows.
\begin{align}
N\leq X&&\implies& C_{({\bf AB,\phi})}=0\label{eq:cabphi1}\\
N>X, K=1&&\implies&C_{({\bf AB,\phi})}=1-\frac{X}{N}\label{eq:cabphi2}\\
2X\geq N > X, \frac{K}{\min(L,M)}\rightarrow\infty&&\implies& C_{({\bf AB,\phi})}\rightarrow 0\label{eq:cabphi3}\\
 N > 2X, \frac{K}{\min(L,M)}\rightarrow\infty&&\implies& C_{({\bf AB,\phi})}\rightarrow 1-\frac{2X}{N}\label{eq:cabphi4}\\
N>X, K\leq\min(L,M),  \frac{\max(L,M)}{K}\rightarrow\infty&&\implies&C_{({\bf AB,\phi})}\rightarrow 1-\frac{X}{N}\label{eq:cabphi5}\\
 2X\geq N>X, &&\implies&C_{({\bf AB,\phi})}\leq \left(1-\frac{X}{N}\right)\frac{\min(L,M,K)}{K}\label{eq:cabphi7}
 \end{align}
 \begin{align}
 N>2X, &&\implies&C_{({\bf AB,\phi})}\leq \left(1-\frac{X}{N}\right)\left(1+\left(\frac{X}{N-X}\right)+\dots+\left(\frac{X}{N-X}\right)^{\lfloor\frac{K}{\min(L,M,K)}\rfloor-1}\right)^{-1} \label{eq:cabphi8}
\end{align}
\end{theorem}}
Case \eqref{eq:cabphi1} is trivial because ${\bf A, B}$ are $X$-secure, and nothing is available to the user as side-information, which means that even if the user and the servers fully combine their knowledge, ${\bf A, B}$ remain a perfect secret.  The converse proof for cases \eqref{eq:cabphi2} and \eqref{eq:cabphi5} of Theorem \ref{thm:cabphi} is presented in Section \ref{sec:confs13}.  Converse proofs for  all other cases follow from 
Lemma \ref{lemma:main} and Theorem \ref{thm:mmxstpir}. For example, consider case \eqref{eq:cabphi3}. According to \eqref{eq:lemmamain1},\eqref{eq:lemmamain2}, we have $C_{({\bf AB,\phi})}\leq C_{\text{MM-XSTPIR}}(N,X,X,K,\min(L,M))$ which in turn is bounded by $\min(L,M)(N-X)/(KN)$ according to Theorem \ref{thm:mmxstpir}. Therefore, if $K/\min(L,M)\rightarrow\infty$, then we have the bound $C_{({\bf AB,\phi})}=0$. Converse bounds for other cases are found similarly.
The proof of achievability for case \eqref{eq:cabphi2} is presented in Section \ref{sec:achfs1}. All other achievability results are presented in Section \ref{sec:achfs2}.

\begin{remark}\label{rem:highlights}
{ 
Let us highlight some special cases of the capacity of SDBMM$_{({\bf AB},\phi)}$ from Theorem \ref{thm:cabphi}.
\begin{enumerate}
  \item Suppose $K=1$, i.e., ${\bf A}_s$ are column vectors and ${\bf B}_s$ are row vectors, so that the desired computations ${\bf A}_s{\bf B}_s$ are outer products. From case \eqref{eq:cabphi2}, the capacity for this setting is $\left(1-\frac{X}{N}\right)^+$. 
  \item Suppose $L=M=1$ and $K\rightarrow\infty$, i.e., ${\bf A}_s$ are  row vectors and ${\bf B}_s$ are  column vectors, so that the desired computations ${\bf A}_s{\bf B}_s$ are inner products of long $(K\rightarrow\infty)$ vectors.   From case \eqref{eq:cabphi4}, the capacity for this setting approaches $\left(1-\frac{2X}{N}\right)^+$ as the length of the vectors ($K$) approaches infinity. 
  \item Suppose $L=M=1$ and $K$ can take arbitrary values, i.e., ${\bf A}_s$ are  row vectors and ${\bf B}_s$ are  column vectors, so that the desired computations ${\bf A}_s{\bf B}_s$ are inner products of arbitrary length vectors.  For $N\leq 2X$, the capacity for this setting is $\frac{1}{K}\left(1-\frac{X}{N}\right)^+$. The converse follows from \eqref{eq:cabphi7}. The achievability is obtained by retrieving each of the $SK$ scalar products contained in the batch of inner products using the scheme in Section \ref{sec:achfs11} and computing the summations at the user locally. Therefore, the rate achieved is $\frac{1}{K}\left(1-\frac{X}{N}\right)^+$, which matches the converse bound.
  \item If the matrix product is not rank deficient, i.e., $K\geq \min(L,M)$, then the capacity is at least $\left(1-\frac{2X}{N}\right)^+$. This follows from the CSA based scheme in Section \ref{sec:csascheme} and Lemma \ref{lemma:Hab}.
  \item From cases \eqref{eq:cabphi7} and \eqref{eq:cabphi8}, when $K>2\min(L,M)$, the capacity of SDBMM$_{({\bf AB},\phi)}$ is strictly smaller than $\left(1-\frac{X}{N}\right)^+$.
\end{enumerate}
}
\end{remark}

\begin{remark}The capacity of $2$-sided SDBMM problem is characterized in \cite{Kakar_Ebadifar_Sezgin_CSA} as $\left(1-\frac{2X}{N}\right)^+$. Our capacity characterizations for cases \eqref{eq:cabphi2} and \eqref{eq:cabphi5} present a contradiction that calls into question\footnote{The information provided by the genie to the user in the converse proof of \cite{Kakar_Ebadifar_Sezgin_CSA} is subsequently considered  useless on the basis that it is independent of ${\bf AB}$. However, it turns out this independent side-information can still be useful in decoding ${\bf AB}$, just as a noise term ${\bf Z}$ that is independent of ${\bf AB}$ can still be useful in decoding ${\bf AB}$ from the  value ${\bf AB}+{\bf Z}$.}  the converse bound in \cite{Kakar_Ebadifar_Sezgin_CSA}. To further highlight the contradiction, note that in the SDBMM$_{\bf AB,\phi}$ problem, for arbitrary $L,K,M$, just by retrieving each of ${\bf A}, {\bf B}$ separately using the scheme described in Section \ref{sec:generalscheme}, it is possible to achieve a rate equal to $\frac{H_q({\bf AB})}{(LK+KM)}\left(1-\frac{X}{N}\right)\geq\frac{\min(LM, LK+KM-K^2)}{(LK+KM)}\left(1-\frac{X}{N}\right)$ which can be larger than $\left(1-\frac{2X}{N}\right)^+$. Since \cite{Kakar_Ebadifar_Sezgin_CSA} assumes $S=1$, consider for example, $L=K=M$, and $N=X+1$. Then  with $S=1$ this simple scheme achieves a rate $\frac{1}{2}\left(1-\frac{X}{N}\right)=\frac{1}{2N}$ which exceeds $\left(1-\frac{2X}{N}\right)^+=0$ for all $X>1$. On the other hand, the achievable scheme presented in \cite{Kakar_Ebadifar_Sezgin_CSA} does not\footnote{This is because $H_q({\bf AB})\neq LM$ when $K<\min(L,M)$. Instead, according to Lemma \ref{lemma:Hab}, $H_q({\bf AB})=LK+KM-M^2$. So while the download for the scheme in \cite{Kakar_Ebadifar_Sezgin_CSA} is $LMN/(N-2X)$, the rate achieved when $K<\min(L,M)$ is $H_q({\bf AB})/D=\frac{LK+KM-K^2}{LM}\left(1-\frac{2X}{N}\right)^+$, which is strictly smaller than $\left(1-\frac{2X}{N}\right)^+$ for $K<\min(L,M)$.} achieve the rate $\left(1-\frac{2X}{N}\right)^+$ when $K<\min(L,M)$, so  it remains unknown if $\left(1-\frac{2X}{N}\right)^+$ is even a lower bound on capacity in general.
\end{remark}

{
\begin{corollary}\label{cor:dimindabphi}
  The dimension-independent capacity of SDBMM$_{({\bf AB,\phi})}$ with $X_A=X_B=X$ is 
  \begin{align}
    \mathring{C}_{({\bf AB,\phi})}=\left\{\begin{aligned}
      &0, && N\leq 2X,\\
      &1-\frac{2X}{N}, && N>2X.
    \end{aligned}\right.
  \end{align}
\end{corollary}
Note that Corollary \ref{cor:dimindabphi} follows immediately from cases \eqref{eq:cabphi1}, \eqref{eq:cabphi3} and \eqref{eq:cabphi4} of Theorem \ref{thm:cabphi}. Also, Corollary \ref{cor:dimindabphi} fully characterizes the dimension-independent capacity of SDBMM$_{({\bf AB,\phi})}$.
}

\begin{corollary}\label{cor:chadabphi}
Consider a modification of the SDBMM$_{({\bf AB,\phi})}$ problem, where  instead of ${\bf AB}$, the user wants to retrieve the Hadamard product ${\bf A}\circ{\bf B}$. The capacity of this problem, i.e., the supremum of $H_q({\bf A}\circ{\bf B})/D$, as $q\rightarrow\infty$, is $1-X/N$.
\end{corollary}
The converse for Corollary \ref{cor:chadabphi} follows directly from the converse proof of Theorem \ref{thm:cabphi},  case \eqref{eq:cabphi2} in Section \ref{sec:confs13} where we replace  ${\bf AB}$ with ${\bf A}\circ{\bf B}$. On the other hand, since the Hadamard product is the entrywise product of matrices, thus the scalar multiplication scheme presented in Section \ref{sec:achfs11} achieves the capacity.

\subsection{Capacity of SDBMM$_{(\bf B,A)}$}
The next SDBMM model we consider corresponds to the one-sided SDBMM problem considered in \cite{Chang_Tandon_SDMMOS}. Recall that the one-sided security model in \cite{Chang_Tandon_SDMMOS} assumes that one of the matrices is a constant matrix known to everyone. This corresponds to the ${\bf A}$ matrix in our model of  SDBMM$_{(\bf B,A)}$ because ${\bf A}$ is not secured and is available to the user as side-information. Here our capacity result is consistent with \cite{Chang_Tandon_SDMMOS}.
\begin{theorem}\label{thm:cba}
The capacity of SDBMM$_{(\bf B,A)}$ with $X_A=0, X_B=X$ is
\begin{align}
C_{(\bf B,A)}=\left\{
\begin{aligned}
&0,&&N\leq X,\\
&1-\frac{X}{N},&&N>X.
\end{aligned}
\right.
\end{align}
\end{theorem}
Theorem \ref{thm:cba} fully characterizes the capacity of SDBMM$_{(\bf B,A)}$. The converse for Theorem \ref{thm:cba} follows along the same lines as the converse presented in \cite{Chang_Tandon_SDMMOS}, but for the sake of completeness we present the converse in Section \ref{sec:concba}. The proof of achievability  provided in  \cite{Chang_Tandon_SDMMOS} is tight only\footnote{ The achievable scheme for the one-sided secure setting in \cite{Chang_Tandon_SDMMOS} always downloads $NLM/(N-X)$ $q$-ary symbols, whereas the capacity achieving scheme needs to download only $NKM/(N-X)$ $q$-ary symbols when $K<L$.} if $K\geq L$ and  $L$ is a multiple of $N-X$. Therefore,  a complete proof of achievability is needed for Theorem \ref{thm:cba}. Such a proof is presented in Section \ref{sec:achcba}. {Also note that the capacity of SDBMM$_{(\bf B,A)}$ is {\it independent of} matrix dimensions $L, K$ and $M$.}

\subsection{Capacity of SDBMM$_{(\bf B,B)}$}
\begin{theorem}\label{thm:cbb}
The capacity of SDBMM$_{({\bf B},{\bf B})}$, with $X_A=0, X_B=X$, is characterized under various settings as follows.
\begin{align}
K\leq M&&\implies& C_{(\bf B,{\bf B})}=1&&~\label{eq:cbb1}\\
K> M, N\leq X&&\implies&C_{(\bf B,B)}=\frac{M}{K}&&\label{eq:cbb2}\\
 N>X, \frac{K}{M}\rightarrow\infty&&\implies&C_{(\bf B,B)}\rightarrow  1-\frac{X}{N}&&\label{eq:cbb3}\\
K>M, N>X&&\implies&C_{(\bf B,B)}\leq \left(1+\left(\frac{X}{N}\right)+\dots+\left(\frac{X}{N}\right)^{\lfloor\frac{K}{M}\rfloor-1}\right)^{-1}&&\label{eq:cbb4}
\end{align}
\end{theorem}

The converse for $K\leq M$ is trivial because the capacity by definition cannot exceed $1$. The converse for the remaining cases  follows directly from Lemma \ref{lemma:main} and Theorem \ref{thm:mmxstpir}. For example, consider the case \eqref{eq:cbb2}. According to Lemma \ref{lemma:main}, C$_{\bf B,B}\leq C_{\text{MM-XSTPIR}}(N,0,X,K,M)$ which is bounded by $M/K$ according to Theorem \ref{thm:mmxstpir}. Other cases follow similarly. The proof of achievability for Theorem \ref{thm:cbb} is provided in Section \ref{sec:achosrsi12}.

{
\begin{corollary}\label{cor:dimindbb}
  The dimension-independent capacity of SDBMM$_{({\bf B,B})}$ with $X_A=0, X_B=X$ is 
  \begin{align}
    \mathring{C}_{({\bf B,B})}=\left\{\begin{aligned}
      &0, && N\leq X,\\
      &1-\frac{X}{N}, && N>X.
    \end{aligned}\right.
  \end{align}
\end{corollary}
Note that Corollary \ref{cor:dimindbb} follows immediately from cases \eqref{eq:cbb2} (set $K/M\rightarrow\infty$) and \eqref{eq:cbb3} of Theorem \ref{thm:cbb}. Besides, Corollary \ref{cor:dimindbb} fully characterizes the dimension-independent capacity of SDBMM$_{({\bf B,B})}$.
}

\subsection{Capacity of SDBMM$_{({\bf B},\phi)}$}
\begin{theorem}\label{thm:cbphi}
The capacity of SDBMM$_{({\bf B},\phi)}$, with $X_A=0, X_B=X$, is characterized under various settings as follows.
\begin{align}
N\leq X&&\implies& C_{(\bf B,\phi)}=0&&~\label{eq:cbphi1}\\
K\geq L, N>X&&\implies&C_{(\bf B,\phi)}=\left(1-\frac{X}{N}\right)\frac{H_q({\bf AB})}{H_q({\bf AB}\mid {\bf A})}=\left(1-\frac{X}{N}\right)&&\label{eq:cbphi2}\\
K<L, N>X, \frac{K}{M}\rightarrow\infty&&\implies&C_{(\bf B,\phi)}\rightarrow 1-\frac{X}{N}&&\label{eq:cbphi3}\\
K\leq M, N>X, \frac{L}{M}\rightarrow\infty&&\implies& C_{(\bf B,\phi)}\rightarrow 1&&\label{eq:cbphi4}\\
K<L, N>X, \frac{M}{L}\rightarrow\infty&&
\implies&C_{(\bf B,\phi)}\rightarrow \left(1-\frac{X}{N}\right)\frac{H_q({\bf AB})}{H_q({\bf AB}\mid {\bf A})}=\left(1-\frac{X}{N}\right)&&\label{eq:cbphi5}\\
L>K\geq M, N>X&&\implies& C_{(\bf B,\phi)}\leq \left(1+\left(\frac{X}{N}\right)+\dots+\left(\frac{X}{N}\right)^{\lfloor\frac{K}{M}\rfloor-1}\right)^{-1}&&\label{eq:cbphi6}
\end{align}
\end{theorem}
The case $N\leq X$ is trivial because  ${\bf B}$ must be $X$-secure and there is no side-information at the user, which means that neither the user, nor all servers together have any knowledge of ${\bf B}$. The converse for \eqref{eq:cbphi2} and \eqref{eq:cbphi5} follows from the fact that any SDBMM$_{\bf B,\phi}$ scheme is also a valid SDBMM$_{\bf B,A}$ scheme, so the  download, say $D_{\bf B,\phi}$ for the best SDBMM$_{\bf B,\phi}$ scheme cannot be less than the  download, say $D_{\bf B,A}$ for the best SDBMM$_{\bf B, A}$ scheme. For \eqref{eq:cbphi3} the converse follows directly from Lemma \ref{lemma:main} and Theorem \ref{thm:mmxstpir}. The converse for \eqref{eq:cbphi4} is trivial because the capacity can never be more than $1$ by definition.
Finally, the converse for \eqref{eq:cbphi6} also follows from Lemma \ref{lemma:main} and Theorem \ref{thm:mmxstpir}. The achievability results for Theorem \ref{thm:cbphi} are proved in Section \ref{sec:achos23}.

{
\begin{corollary}\label{cor:dimindbphi}
  The dimension-independent capacity of SDBMM$_{({\bf B,\phi})}$ with $X_A=0, X_B=X$ is 
  \begin{align}
    \mathring{C}_{({\bf B,\phi})}=\left\{\begin{aligned}
      &0, && N\leq X,\\
      &1-\frac{X}{N}, && N>X.
    \end{aligned}\right.
  \end{align}
\end{corollary}
Note that Corollary \ref{cor:dimindbphi} follows immediately from cases \eqref{eq:cbphi1}, \eqref{eq:cbphi2}, \eqref{eq:cbphi3} and \eqref{eq:cbphi5} of Theorem \ref{thm:cbphi}. Corollary \ref{cor:dimindbphi} fully characterizes the dimension-independent capacity of SDBMM$_{({\bf B,\phi})}$.
}

\subsection{Capacity of SDBMM$_{{\bf AB},{\bf B}}$}
\begin{theorem}\label{thm:cabb}
The capacity of SDBMM$_{({\bf AB},{\bf B})}$, with $X_A=X_B=X$, is characterized under various settings as follows.
\begin{align}
N\leq X&&\implies& C_{({\bf AB,B})}=0\label{eq:cabb1}\\
N>X, K\leq M&&\implies&C_{({\bf AB,B})}=1-\frac{X}{N}\label{eq:cabb2}\\
2X\geq N > X, K>M&&\implies& C_{({\bf AB,B})}=\frac{M(N-X)}{KN}\label{eq:cabb3}\\
N>2X, K>M, \frac{K}{M}\rightarrow\infty&&\implies&C_{({\bf AB,B})}\rightarrow 1-\frac{2X}{N}\label{eq:cabb4}\\
N>2X, K>M&&\implies &C_{({\bf AB,B})}\leq \frac{N-X}{N}\left(1+\left(\frac{X}{N-X}\right)+\dots+\left(\frac{X}{N-X}\right)^{\lfloor\frac{K}{M}\rfloor-1}\right)^{-1}\label{eq:cabb5}
\end{align}
\end{theorem}
The case \eqref{eq:cabb1} with $N\leq X$ is trivial because the ${\bf A}$ is $X$-secure and not available to the user as side-information, which means that it is unknown to both the user and all servers. The converse for \eqref{eq:cabb2} follows from the observation that relaxing the security constraint for ${\bf B}$ cannot hurt, so $C_{({\bf AB,B})}$ is bounded above by $C_{({\bf A,B})}$, which is equal to $C_{({\bf B,A})}=1-\frac{X}{N}$ by the symmetry of the problem and the result of Theorem \ref{thm:cba}. The converse for \eqref{eq:cabb3}, \eqref{eq:cabb4}, \eqref{eq:cabb5} follows from Lemma \ref{lemma:main} and Theorem \ref{thm:mmxstpir}.
The proof of achievability  for Theorem \ref{thm:cabb} appears in Section \ref{sec:achfssi12}.

{
\begin{corollary}\label{cor:dimindabb}
  The dimension-independent capacity of SDBMM$_{({\bf AB,B})}$ with $X_A=X_B=X$ is 
  \begin{align}
    \mathring{C}_{({\bf AB,B})}=\left\{\begin{aligned}
      &0, && N\leq 2X,\\
      &1-\frac{2X}{N}, && N>2X.
    \end{aligned}\right.
  \end{align}
\end{corollary}
Note that Corollary \ref{cor:dimindabb} follows immediately from cases \eqref{eq:cabb1}, \eqref{eq:cabb3} (set $K/M\rightarrow\infty$) and \eqref{eq:cabb4} of Theorem \ref{thm:cabb}. Besides, Corollary \ref{cor:dimindabb} fully characterizes the dimension-independent capacity of SDBMM$_{({\bf AB,B})}$.
}

\section{Converse}
\subsection{Proof of Converse for Theorem \ref{thm:cba}}\label{sec:concba}
The case $N\leq X$ is trivial because  ${\bf B}$ must be $X$-secure and not available to the user as side-information, which means that neither the user, nor all servers together have any knowledge of ${\bf B}$. Now let us consider the case $N>X$. Let $\mathcal{X}$ denote any subset of $[N]$ such that $|\mathcal{X}|=X$. We start with the following lemma.
\begin{lemma}\label{lemma:conoslsi1}
$I_q\left(\Delta_{\mathcal{X}};\mathbf{AB}\mid\mathbf{A}\right)=0$.
\end{lemma}
\begin{proof}
\begin{align}
&I_q\left(\Delta_{\mathcal{X}};\mathbf{AB}\mid\mathbf{A}\right)\notag\\
&=H_q(\Delta_{\mathcal{X}}\mid\mathbf{A})-H_q\left(\Delta_{\mathcal{X}}\mid \mathbf{AB},\mathbf{A}\right)\label{eq:lmconoslsi1}\\
&\leq H_q(\Delta_{\mathcal{X}}\mid\mathbf{A})-H_q\left(\Delta_{\mathcal{X}}\mid\mathbf{A},\mathbf{B}\right)\label{eq:lmconoslsi2}\\
&=I_q\left(\Delta_{\mathcal{X}};\mathbf{B}\mid\mathbf{A}\right)\label{eq:lmconoslsi3}\\
&\leq I_q\left(\widetilde{B}_{[S]}^{\mathcal{X}},\mathbf{A};\mathbf{B}\mid\mathbf{A}\right)\label{eq:lmconoslsi4}\\
&\leq I_q\left(\widetilde{B}_{[S]}^{\mathcal{X}},\mathbf{A};\mathbf{B}\right)\label{eq:lmconoslsi5}\\
&= I_q\left(\widetilde{B}_{[S]}^{\mathcal{X}}; {\bf B}\right)+I_q\left(\mathbf{A};{\bf B}\mid \widetilde{B}_{[S]}^{\mathcal{X}}\right)\label{eq:lmconoslsi6}\\
&\leq I_q\left(\widetilde{B}_{[S]}^{\mathcal{X}}; {\bf B}\right)+I_q\left(\mathbf{A};\mathbf{B},\widetilde{B}_{[S]}^{\mathcal{T}}\right)\label{eq:lmconoslsi7}\\
&=0.\label{eq:lmconoslsi8}
\end{align}
The steps in the proof are justified as follows. Step \eqref{eq:lmconoslsi1} applies the definition of mutual information. Step \eqref{eq:lmconoslsi2} follows from the fact that $\left({\bf AB},\mathbf{A}\right)$ is function of $\left(\mathbf{A},\mathbf{B}\right)$ and conditioning reduces entropy. Step \eqref{eq:lmconoslsi3} applies the definition of mutual information, and \eqref{eq:lmconoslsi4} holds because $\Delta_{\mathcal{X}}$ is function of $\left(\widetilde{B}_{[S]}^{\mathcal{T}},\mathbf{A}\right)$. In \eqref{eq:lmconoslsi5}, \eqref{eq:lmconoslsi6} and \eqref{eq:lmconoslsi7}, we repeatedly used the chain rule and non-negativity of mutual information. The last step follows from the security constraint defined in \eqref{eq:secureb}  and separate encoding of matrices \eqref{eq:sepencode}. The proof is  completed by the non-negativity of mutual information.
\end{proof}

The proof of converse of Theorem \ref{thm:cba} is now presented as follows.
\begin{align}
&H_q\left({\bf AB}\mid{\bf A}\right)\notag\\
&=H_q\left({\bf AB}\mid{\bf A}\right)-H_q\left({\bf AB}|\Delta_{[N]},\mathbf{A}\right)+H_q\left({\bf AB}\mid\Delta_{[N]},\mathbf{A}\right)\label{eq:conoslsi1}\\
&=H_q\left({\bf AB}\mid{\bf A}\right)-H_q\left({\bf AB}\mid\Delta_{[N]},\mathbf{A}\right)\label{eq:conoslsi2}\\
&=I_q\left({\bf AB};\Delta_{[N]}\mid \mathbf{A}\right)\label{eq:conoslsi3}\\
&=H_q(\Delta_{[N]}\mid\mathbf{A})-H_q\left(\Delta_{[N]}\mid{\bf AB},\mathbf{A}\right)\label{eq:conoslsi4}\\
&\leq H_q(\Delta_{[N]}\mid \mathbf{A})-H_q\left(\Delta_{\mathcal{X}}\mid \mathbf{AB},\mathbf{A}\right)\label{eq:conoslsi5}\\
&=H_q(\Delta_{[N]}\mid\mathbf{A})-H_q\left(\Delta_{\mathcal{X}}\mid\mathbf{A}\right).\label{eq:conoslsi6}
\end{align}
Steps are justified as follows. \eqref{eq:conoslsi1} subtracts and adds the same term so nothing changes. \eqref{eq:conoslsi2} follows from the correctness constraint,\eqref{eq:correct}. Steps \eqref{eq:conoslsi3}  and \eqref{eq:conoslsi4} follow from the definition of mutual information. In \eqref{eq:conoslsi5}, we used the fact that dropping terms reduces entropy. The last step holds from Lemma \ref{lemma:conoslsi1}.

Averaging \eqref{eq:conoslsi6} over all choices of $\mathcal{X}$ and applying Han's inequality (Theorem 17.6.1 in \cite{Cover_Thomas}), we have
\begin{align}
&H_q\left(\mathbf{AB}\mid\mathbf{A}\right)\notag\\
&\leq H_q(\Delta_{[N]}\mid \mathbf{A})-\frac{X}{N}H_q\left(\Delta_{[N]}\mid\mathbf{A}\right)\\
&=\left(1-\frac{X}{N}\right)H_q(\Delta_{[N]}|\mathbf{A})\\
&\leq\left(1-\frac{X}{N}\right)H_q(\Delta_{[N]})\\
&\leq \left(1-\frac{X}{N}\right)\sum_{n\in[N]}H_q(\Delta_{n}).
\end{align}
Thus we obtain
\begin{align}
C_{\bf B,A}&=\sup \frac{H_q\left(\mathbf{AB}\mid\mathbf{A}\right)}{D}\\
&\leq\sup \frac{H_q\left(\mathbf{AB}\mid \mathbf{A}\right)}{\sum_{n\in[N]}H_q(\Delta_{n})}\\
&\leq 1-\frac{X}{N}.
\end{align}

\subsection{Converse of Theorem \ref{thm:cabphi}: \eqref{eq:cabphi2},\eqref{eq:cabphi5}}\label{sec:confs13}
Let $\mathcal{X}$ denote any subset of $[N]$ such that $|\mathcal{X}|=X$. The proof of converse is as follows.
\begin{align}
&H_q({\bf AB})\notag\\
&=H_q({\bf AB})-H_q\left({\bf AB}\mid \Delta_{[N]}\right)+H_q\left({\bf AB}\mid \Delta_{[N]}\right)\label{eq:confssi1}\\
&=H_q({\bf AB})-H_q\left({\bf AB}\mid \Delta_{[N]}\right)\label{eq:confssi2}\\
&=I_q\left({\bf AB};\Delta_{[N]}\right)\label{eq:confssi3}\\
&=H_q\left(\Delta_{[N]}\right)-H_q\left(\Delta_{[N]}\mid {\bf AB}\right)\label{eq:confssi4}\\
&\leq H_q(\Delta_{[N]})-H_q\left(\Delta_{[N]}\mid {\bf A,B}\right)\label{eq:confssi5}\\
&\leq H_q(\Delta_{[N]})-H_q\left(\Delta_{\mathcal{X}}\mid\mathbf{A},\mathbf{B}\right)\label{eq:confssi6}\\
&=H_q(\Delta_{[N]})-H_q\left(\Delta_{\mathcal{X}}\right).\label{eq:confssi7}
\end{align}
Steps are justified as follows. \eqref{eq:confssi1} subtracts and adds the same term so nothing changes. \eqref{eq:confssi2} follows from \eqref{eq:correct}, while \eqref{eq:confssi3} and \eqref{eq:confssi4} follow from the definition of mutual information. \eqref{eq:confssi5} holds because adding conditioning reduces entropy and ${\bf AB}$ is function of $(\mathbf{A},\mathbf{B})$. \eqref{eq:confssi6} holds because dropping terms reduces entropy. The last step simply follows from the following fact,
\begin{align}
0&=I_q(\mathbf{A}, \mathbf{B};\widetilde{A}_{[S]}^{\mathcal{X}}, \widetilde{B}_{[S]}^{\mathcal{X}})\label{eq:confssi8}\\
&=I_q(\mathbf{A}, \mathbf{B};\Delta_{\mathcal{X}})\label{eq:confssi9}\\
&=H_q(\Delta_{\mathcal{X}})-H_q(\Delta_{\mathcal{X}}\mid \mathbf{A}, \mathbf{B})\label{eq:confssi10}
\end{align}
where \eqref{eq:confssi8} is the security constraint defined in \eqref{eq:secure}. \eqref{eq:confssi9} follows from non-negativity of mutual information and the fact that $\Delta_{\mathcal{X}}$ is function of $(\widetilde{A}_{[S]}^{\mathcal{X}}, \widetilde{B}_{[S]}^{\mathcal{X}})$. \eqref{eq:confssi10} is the definition of mutual information.

Averaging \eqref{eq:confssi7} over all choices of $\mathcal{X}$ and applying Han's inequality, we have
\begin{align}
H_q\left({\bf AB}\right)&\leq H_q\left(\Delta_{[N]}\right)-\frac{X}{N}H_q\left(\Delta_{[N]}\right)\\
&=\left(1-\frac{X}{N}\right)H_q\left(\Delta_{[N]}\right)\\
&\leq\left(1-\frac{X}{N}\right)\sum_{n\in[N]}H_q(\Delta_{n}).
\end{align}
Thus we obtain
\begin{align}
C_{({\bf AB,\phi})}&=\sup\frac{H_q({\bf AB})}{D}\\
&\leq \sup\frac{H_q({\bf AB})}{\sum_{n\in[N]}H_q(\Delta_{n})}\\
&\leq 1-\frac{X}{N}.
\end{align}

\section{Achievability}
Let us present two  basic schemes that are essential ingredients of the proofs of achievability.
\subsection{A General Scheme}\label{sec:generalscheme}
This scheme allows the user to retrieve all ${\bf A}, {\bf B}$, after which  he can locally compute ${\bf AB}$.
Let $S=N-X$, and let $\mathbf{Z}_{sx}, \mathbf{Z}'_{sx'}$, $s\in[S], x\in[X_A], x'\in[X_B]$ be uniformly distributed random matrices over $\mathbb{F}_q^{L\times K}$ and $\mathbb{F}_q^{K\times M}$ respectively. Note that $X_A,X_B\in\{0,X\}$. The independence of these random matrices and matrices $\mathbf{A}_{[S]}, \mathbf{B}_{[S]}$ is specified as follows.
\begin{align}
&H_q\left((\mathbf{Z}_{sx},\mathbf{Z}'_{sx'})_{s\in[S],x\in[X_A],x'\in[X_B]},\mathbf{A}_{[S]}, \mathbf{B}_{[S]}\right)\notag\\
=&\sum_{s\in[S],x\in[X_A],x'\in[X_B]}H_q(\mathbf{Z}_{sx})+H_q(\mathbf{Z}'_{sx'})+\sum_{s\in[S]}H_q(\mathbf{A}_{s})+H_q(\mathbf{B}_{s}).
\end{align}
Let $\alpha_n, n\in[N]$ be  $N$ distinct elements from $\mathbb{F}_q$. The construction of securely encoded matrices $\widetilde{A}_s^n$ and $\widetilde{B}_s^n$ for any $s\in[S]$ and $n\in[N]$ is provided below.
\begin{align}
\widetilde{A}_s^n&=\alpha_n^s\mathbf{A}_s+\sum_{x\in[X_A]}\alpha_n^{S+x}\mathbf{Z}_{xt}\\
&=\alpha_n^s\mathbf{A}_s+\alpha_n^{S+1}\mathbf{Z}_{s1}+\dots+\alpha_n^{S+X_A}\mathbf{Z}_{sX_A}\\
\widetilde{B}_s^n&=\alpha_n^s\mathbf{B}_s+\sum_{x'\in[T_B]}\alpha_n^{S+x'}\mathbf{Z}'_{sx'}\\
&=\alpha_n^s\mathbf{B}_s+\alpha_n^{S+1}\mathbf{Z}'_{s1}+\dots+\alpha_n^{S+X_B}\mathbf{Z}'_{sX_B}.
\end{align}
The answer from the $n$-th server is specified as follows
\begin{align}
\Delta_n&=
\begin{bmatrix}
\widetilde{A}_1^n+\dots+\widetilde{A}_S^n\\
\widetilde{B}_1^n+\dots+\widetilde{B}_S^n
\end{bmatrix}\\
&=\begin{bmatrix}
\alpha_n\mathbf{A}_1+\dots+\alpha_n^S\mathbf{A}_S+\alpha_n^{S+1}\sum_{s\in[S]}\mathbf{Z}_{s1}+\dots+\alpha_n^{S+X_A}\sum_{s\in[S]}\mathbf{Z}_{sX_A}\\
\alpha_n\mathbf{B}_1+\dots+\alpha_n^S\mathbf{B}_S+\alpha_n^{S+1}\sum_{s\in[S]}\mathbf{Z}'_{s1}+\dots+\alpha_n^{S+X_B}\sum_{s\in[S]}\mathbf{Z}'_{sX_B}
\end{bmatrix}.
\end{align}
Note that the desired matrices and the random matrices are coded with an RS code. Therefore, from the answers provided by all $N$ servers, the user is able to decode all matrices $\mathbf{A}_{[S]},\mathbf{B}_{[S]}$, and then determine $\mathbf{AB}=(\mathbf{A}_s\times\mathbf{B}_s)_{s\in[S]}$. Note that $X_A$-security is guaranteed for matrices $\mathbf{A}_s$ because they are protected by the $X_A$ noise matrices ${\bf Z}_{sx}, x\in[X_A]$, that are i.i.d. uniform and coded according to $\text{MDS}(X_A,N)$. Similarly, $X_B$-security is guaranteed for matrices $\mathbf{B}_s$.

\subsection{Cross Subspace Alignment Based Scheme}\label{sec:csascheme}
For this scheme, let us set
\begin{equation}
S=N-X_A-X_B.
\end{equation}
And let $\mathbf{Z}_{sx}, \mathbf{Z}'_{sx'}$, $s\in[1:S], x\in[T_A], x'\in[T_B]$ be uniformly distributed random matrices over $\mathbb{F}_q^{L\times K}$ and $\mathbb{F}_q^{K\times M}$ respectively. The independence of random matrices and matrices $\mathbf{A}_{[S]}, \mathbf{B}_{[S]}$ is specified as follows.
\begin{align}
&H_q\left((\mathbf{Z}_{sx},\mathbf{Z}'_{sx'})_{s\in[S],x\in[T_A],x'\in[T_B]},\mathbf{A}_{[S]}, \mathbf{B}_{[S]}\right)\notag\\
=&\sum_{s\in[S],x\in[X_A]}H_q(\mathbf{Z}_{sx})+\sum_{s\in[S],x'\in[X_B]}H_q(\mathbf{Z}'_{sx'})+\sum_{s\in[S]}H_q(\mathbf{A}_{s})+\sum_{s\in[S]}H_q(\mathbf{B}_{s}).
\end{align}
For the construction of this scheme, we will need  $S$ distinct constants $f_s\in\mathbb{F}_q, s\in[S]$, and  $N$ distinct constants $\alpha_n, n\in[N]$ that are elements\footnote{As in the construction of CSA codes in the subsequent work \cite{Jia_Jafar_CDBC}, it is equivalent to choose $S+N$ arbitrary distinct values $(\alpha'_1,\alpha'_2,\cdots,\alpha'_N,f'_1,f'_2,\cdots,f'_S)$ from $\mathbb{F}_q$, and set $\alpha_n=\alpha'_n, \forall n\in[N]$, $f_s=-f'_s, \forall s\in[S]$.} of $\mathbb{G}$,
\begin{equation}\label{eq:distinct}
\mathbb{G}=\{\alpha\in\mathbb{F}_q:\alpha+f_s\neq0, \forall s\in[S]\}.
\end{equation}
The securely encoded matrix $\widetilde{A}_s^n$ for any $s\in[S]$ and $n\in[N]$ is provided below.
\begin{align}
\widetilde{A}_s^n&=\mathbf{A}_s+\sum_{x\in[X_A]}(f_s+\alpha_n)^x\mathbf{Z}_{sx}\\
&=\mathbf{A}_s+(f_s+\alpha_n)\mathbf{Z}_{s1}+\dots+(f_s+\alpha_n)^{X_A}\mathbf{Z}_{sX_A}.
\end{align}
Similarly, the securely encoded matrix $\widetilde{B}_s^n$ for any $s\in[S]$ and $n\in[N]$ is as follows,
\begin{align}
\widetilde{B}_s^n&=\frac{1}{f_s+\alpha_n}\left(\mathbf{B}_s+\sum_{x'\in[X_B]}(f_s+\alpha_n)^{x'}\mathbf{Z}'_{sx'}\right)\\
&=\frac{1}{f_s+\alpha_n}\left(\mathbf{B}_s+(f_s+\alpha_n)\mathbf{Z}'_{s1}+\dots+(f_s+\alpha_n)^{X_B}\mathbf{Z}'_{sX_B}\right).
\end{align}
The download from  any server $n$, $n\in[N]$ is constructed as follows.
\begin{align}
\Delta_n&=\sum_{s\in[S]}\widetilde{A}_s^n\widetilde{B}_s^n\\
&=\sum_{s\in[S]}\left(\frac{1}{f_s+\alpha_n}\mathbf{A}_s\mathbf{B}_s\right)+\sum_{s\in[S]}\sum_{x\in[X_A]}(f_s+\alpha_n)^{x-1}\mathbf{Z}_{sx}\mathbf{B}_s\notag\\
&\quad+\sum_{s\in[S]}\sum_{x'\in[X_B]}(f_s+\alpha_n)^{x'-1}\mathbf{A}_s\mathbf{Z}'_{sx'}+\sum_{s\in[S]}\sum_{x\in[X_A],x'\in[X_B]}(f_s+\alpha_n)^{x+x'-1}\mathbf{Z}_{sx}\mathbf{Z}'_{sx'}.
\end{align}
Each of the last three terms can be expanded into weighted sums of terms of the form $\alpha_n^t$, ${t\in\{0,1,\dots,X_A+X_B-1\}}$. Thus, upon receiving all $N$ answers from servers, the user is able to decode all $S$ desired product matrices $(\mathbf{A}_s\times\mathbf{B}_s)_{s\in[S]}$, as long as the following $N\times N$ matrix is invertible,
\begin{align}
\mathbf{M}_N=\begin{bmatrix}
\frac{1}{f_1+\alpha_1}&\cdots&\frac{1}{f_S+\alpha_1}&1&\alpha_1&\cdots&\alpha_1^{X_A+X_B-1}\\
\frac{1}{f_1+\alpha_2}&\cdots&\frac{1}{f_S+\alpha_2}&1&\alpha_2&\cdots&\alpha_2^{X_A+X_B-1}\\
\vdots&\vdots&\vdots&\vdots&\vdots&\vdots&\vdots&\\
\frac{1}{f_1+\alpha_N}&\cdots&\frac{1}{f_S+\alpha_N}&1&\alpha_N&\cdots&\alpha_N^{X_A+X_B-1}\\
\end{bmatrix},
\end{align}
which is shown to be true by Lemma \ref{lemma:mninv}. The proof of Lemma \ref{lemma:mninv},  which is based on the proof of Lemma 5 in \cite{Jia_Sun_Jafar_XSTPIR}, is presented in Appendix \ref{app:mn}. $X_A$-security is guaranteed for matrices $\mathbf{A}_s$ because they are protected by the $X_A$ noise matrices ${\bf Z}_{sx}, x\in[X_A]$, that are i.i.d. uniform and coded according to $\text{MDS}(X_A,N)$ codes. $X_B$-security is similarly guaranteed for the matrices $\mathbf{B}_s, s\in[S]$.

\subsection{Proofs of Achievability}
Throughout these proofs, we will allow $q\rightarrow\infty$. Furthermore, we will use Lemma \ref{lemma:Hab} to calculate the entropy of random matrices.
\subsubsection{Achievability Proof of Theorem \ref{thm:cba}}\label{sec:achcba}
First, let us consider the setting when $K\geq L$. For this setting, let us apply the cross subspace alignment based scheme presented in Section \ref{sec:csascheme}. Note that $X_A=0, X_B=X$, and the total number of downloaded $q$-ary symbols is $NLM$, so the rate achieved is
\begin{align}
R&=\frac{H_q\left(\mathbf{AB}\mid \mathbf{A}\right)}{NLM}\\
&=\frac{SLM}{NLM}\\
&=1-\frac{X}{N},
\end{align}
which matches the capacity for this setting. On the other hand, when $K<L$, let us apply the general scheme presented in Section \ref{sec:generalscheme}. Since $X_A=0, X_B=X$, we have 
\begin{align}
\Delta_n&=\begin{bmatrix}
\alpha_n\mathbf{A}_1+\dots+\alpha_n^S\mathbf{A}_S\\
\alpha_n\mathbf{B}_1+\dots+\alpha_n^S\mathbf{B}_S+\alpha_n^{S+1}\sum_{s\in[S]}\mathbf{Z}'_{s1}+\dots+\alpha_n^{S+X}\sum_{s\in[S]}\mathbf{Z}'_{sX}
\end{bmatrix}.
\end{align}
But note that left matrices $\mathbf{A}_{[S]}$ are already available to  the user as side information, so it is not necessary to download $(\alpha_n\mathbf{A}_1+\dots+\alpha_n^S\mathbf{A}_S)$ terms. Therefore, the total number of downloaded $q$-ary symbols is $NKM$, and the rate achieved is
\begin{align}
R=&\frac{H_q\left({\bf AB}\mid {\bf A}\right)}{NKM}\\
=&\frac{SKM}{NKM}\\
=&1-\frac{X}{N},
\end{align}
which matches the capacity for this setting. This completes the achievability proof of Theorem \ref{thm:cba}.

\subsubsection{Achievability Proof for Theorem \ref{thm:cbb}}\label{sec:achosrsi12}
First consider the trivial scheme with $S=1$ that downloads the matrices $\mathbf{A}_{[S]}$ directly from any one out of $N$ servers, since there is no security constraint on these matrices $(X_A=0)$. Since $\mathbf{B}_{[S]}$ is already available as side information, downloading $\mathbf{A}_{[S]}$ allows the user to compute $\mathbf{AB}$ locally. The rate achieved with this scheme is
\begin{align}
R&=\frac{H_q\left(\mathbf{AB}\mid \mathbf{B}\right)}{LK}\\
&=\frac{\min(LM,LK)}{LK}\\
&=\left\{\begin{aligned}
&1,\quad&&K\leq M,\\
&\frac{M}{K},\quad&&K>M.
\end{aligned}\right.
\end{align}
Thus, this simple scheme is optimal  for $K\leq M$ and for $(K>M, N\leq X)$. 

Next let us consider $N>X$ as $K/M\rightarrow\infty$. For this, let us apply the cross subspace alignment based scheme presented in Section \ref{sec:csascheme} with $S=N-X$. Note that $X_A=0, X_B=X$, and the total number of downloaded $q$-ary symbols is $NLM$, so the rate achieved is
\begin{align}
R=&\frac{H_q\left({\bf AB}\mid \mathbf{B}\right)}{NLM}\\
=&\frac{SLM}{NLM}\\
=&1-\frac{X}{N},
\end{align}
which matches the capacity for this setting. This completes the achievability proof of Theorem \ref{thm:cbb}.

\subsubsection{Achievability Proof for Theorem \ref{thm:cbphi} }\label{sec:achos23}
For the cases \eqref{eq:cbphi2}, \eqref{eq:cbphi3}, let us apply the cross subspace alignment based scheme presented in Section \ref{sec:csascheme} with $S=N-X$. Since $X_A=0, X_B=X$, and the total number of downloaded $q$-ary symbol is $NLM$, the rate achieved is
\begin{align}
R&=\frac{H_q\left({\bf AB}\right)}{NLM}\\
&=\frac{SLM}{NLM}\\
&=1-\frac{X}{N}.
\end{align}
This completes the achievability proof of Theorem \ref{thm:cbphi} for the cases \eqref{eq:cbphi2}, \eqref{eq:cbphi3}.

Now consider the cases \eqref{eq:cbphi4} and \eqref{eq:cbphi5}. For these cases, let us apply the general scheme presented in Section \ref{sec:generalscheme}. Since $X_A=0, X_B=X$, we have 
\begin{align}
\Delta_n&=\begin{bmatrix}
\alpha_n\mathbf{A}_1+\dots+\alpha_n^S\mathbf{A}_S\\
\alpha_n\mathbf{B}_1+\dots+\alpha_n^S\mathbf{B}_S+\alpha_n^{S+1}\sum_{s\in[S]}\mathbf{Z}'_{s1}+\dots+\alpha_n^{S+X}\sum_{s\in[S]}\mathbf{Z}'_{sX}
\end{bmatrix}.
\end{align}
Note that from the downloads $\Delta_n$ of any $S$ servers, we are able to recover the matrices $\mathbf{A}_{[S]}$, so we can eliminate the first part from the remaining $N-S$ redundant downloads while preserving decodability. Therefore, the total number of downloaded $q$-ary symbols is $SLK+NKM$. Thus, as $q\rightarrow\infty$, the rate achieved is
\begin{align}
R&=\frac{H_q\left({\bf AB}\right)}{SLK+NKM}\\
&=\frac{S(LK+KM-K^2)}{SLK+NKM}
\end{align}
As $L/M\rightarrow\infty$ and when $M\geq K$, we have $R=1$. This completes the proof of achievability of \eqref{eq:cbphi4}. On the other hand, when $M/L\rightarrow\infty$ and $K<L$, we have
\begin{eqnarray}
R&=&\frac{S(LK+KM-K^2)}{SLK+NKM}\\
&\overset{M/L\rightarrow\infty}{=}&\frac{S}{N}\\
&{=}&1-\frac{X}{N}
\end{eqnarray}
This proves achievability for \eqref{eq:cbphi5}, thus completing the proof of achievability for Theorem \ref{thm:cbphi}.

\subsubsection{Achievability Proof of Theorem \ref{thm:cabb}}\label{sec:achfssi12}
Let us start with the cases \eqref{eq:cabb2} and \eqref{eq:cabb3}, for which we apply the general scheme presented in Section \ref{sec:generalscheme}. Since $X_A=X_B=X$, we have 
\begin{align}
\Delta_n&=\begin{bmatrix}
\alpha_n\mathbf{A}_1+\dots+\alpha_n^S\mathbf{A}_S+\alpha_n^{S+1}\sum_{s\in[S]}\mathbf{Z}_{s1}+\dots+\alpha_n^{S+X}\sum_{s\in[S]}\mathbf{Z}_{sX}\\
\alpha_n\mathbf{B}_1+\dots+\alpha_n^S\mathbf{B}_S+\alpha_n^{S+1}\sum_{s\in[S]}\mathbf{Z}'_{s1}+\dots+\alpha_n^{S+X}\sum_{s\in[S]}\mathbf{Z}'_{sX}
\end{bmatrix}.
\end{align}
Now we note that since the matrices $\mathbf{B}_{[S]}$ are available to user as side information, it is not necessary to download the second part of $\Delta_n$. Therefore, the total number of downloaded $q$-ary symbols is $NLK$, and the rate achieved is
\begin{align}
R&=\frac{H_q\left({\bf AB}\mid {\bf B}\right)}{NLK}\\
&=\frac{S\min(LK,LM)}{NLK}\\
&=\left\{\begin{aligned}
&1-\frac{X}{N},\quad&&K\leq M\\
&\frac{M(N-X)}{KN},\quad&&K>M
\end{aligned}\right..
\end{align}
This completes the achievability proof for cases \eqref{eq:cabb2} and \eqref{eq:cabb3}.

Next, let us consider case \eqref{eq:cabb4}, and for this setting let us apply the cross subspace alignment based scheme presented in Section \ref{sec:csascheme}. Note that $X_A=X_B=X$, and the total number of downloaded $q$-ary symbols is $NLM$, so the rate achieved is
\begin{align}
R&=\frac{H_q\left({\bf AB}\mid {\bf B}\right)}{NLM}\\
&=\frac{SLM}{NLM}\\
&=1-\frac{2X}{N}
\end{align}
which matches the capacity for this setting.

\subsubsection{Achievability Proof of Theorem \ref{thm:cabphi} }\label{sec:achfs2}
Let us start with case \eqref{eq:cabphi4}, for which we apply the cross subspace alignment based scheme that was presented in Section \ref{sec:csascheme}.  Note that $X_A=X_B=X$, and the total number of downloaded $q$-ary symbols is $NLM$, so the rate achieved is 
\begin{align}
R&=\frac{H_q\left({\bf AB}\right)}{NLM}\\
&=\frac{SLM}{NLM}\\
&=1-\frac{2X}{N},
\end{align}
which matches the capacity for this setting.

Next, consider case \eqref{eq:cabphi5}. 
For this setting, let us apply the general scheme presented in Section \ref{sec:generalscheme}. Since $X_A=X_B=X$, we have 
\begin{align}
\Delta_n&=\begin{bmatrix}
\alpha_n\mathbf{A}_1+\dots+\alpha_n^S\mathbf{A}_S+\alpha_n^{S+1}\sum_{s\in[S]}\mathbf{Z}_{s1}+\dots+\alpha_n^{S+X}\sum_{s\in[S]}\mathbf{Z}_{sX}\\
\alpha_n\mathbf{B}_1+\dots+\alpha_n^S\mathbf{B}_S+\alpha_n^{S+1}\sum_{s\in[S]}\mathbf{Z}'_{s1}+\dots+\alpha_n^{S+X}\sum_{s\in[S]}\mathbf{Z}'_{sX}
\end{bmatrix}.
\end{align}
Thus the total number of downloaded $q$-ary symbols is $N(LK+KM)$. Therefore,  when $K\leq \min(L,M)$, we have
\begin{eqnarray}
R&=&\frac{H_q\left({\bf AB}\right)}{N(LK+KM)}\\
&=&\frac{S(LK+KM-K^2)}{N(LK+KM)}\\
&\overset{\max(L,M)/K\rightarrow\infty}{=}&1-\frac{X}{N}.
\end{eqnarray}
This completes the achievability proof of case \eqref{eq:cabphi5}.

\subsection{Achievability Proof of Theorem \ref{thm:cabphi}: Case \eqref{eq:cabphi2}}\label{sec:achfs1}
\subsubsection{$K=L=M=1$}\label{sec:achfs11}
Let us first consider the setting where $K=L=M=1$, and let us set $S=N-X$. Note that in this setting, $\mathbf{A}_{s}$, $\mathbf{B}_{s}$ are  independent scalars drawn uniformly from the finite field $\mathbb{F}_q$. Let us first present a solution based on the assumption that  $\mathbf{A}_{s}$, $\mathbf{B}_{s}$ take only non-zero values for all $s\in[S]$. It is well-known that the multiplicative group $\mathbb{F}_q^{\times}=\mathbb{F}_q\setminus\{0\}$ is a cyclic group. Moreover, every finite cyclic group of order $q-1$ is isomorphic to the additive group of $\mathbb{Z}/(q-1)\mathbb{Z}$ (i.e., addition modulo $(q-1)$). Therefore it is possible to translate scalar multiplication over $\mathbb{F}_q^{\times}$ into addition modulo $(q-1)$. However, the additive group of $\mathbb{Z}/(q-1)\mathbb{Z}$ is not a field, and our scheme will further require the properties of a field. This problem is circumvented by using a prime field $\mathbb{F}_p$ for a prime $p$ such that $p>2(q-1)$ and noting that for any two integers $a,b\in\{0,1,\dots,q-2\}$, we have
\begin{align}
(a+b)\mod(q-1)=((a+b)\mod p)\mod(q-1).
\end{align}
In other words, suppose the isomorphism between the multiplicative group $\mathbb{F}_q^\times$ and the additive group  $\mathbb{Z}/(q-1)\mathbb{Z}$ maps all $a\in\mathbb{F}_q^\times$ to $f(a)\in\mathbb{Z}/(q-1)\mathbb{Z}$. Then for all $a,b,c\in\mathbb{F}_q^\times$ such that $c=a\times b$, we have $f(c)=f(a)+f(b)$ in $\mathbb{Z}/(q-1)\mathbb{Z}$, and furthermore, under the natural interpretation of all $f(a)$ as elements of $\mathbb{F}_p$, we have $c'=f(a)+f(b)$ in $\mathbb{F}_p$ such that $f(c)=c'$ mod $(q-1)$.
Thus, we are able to transform the problem of scalar multiplication in $\mathbb{F}_q^\times$ to scalar addition over $\mathbb{F}_p$, i.e., instead of $c=a\times b\in \mathbb{F}_q^\times$, the user will retrieve $c'=f(a)+f(b)\in\mathbb{F}_p$ from which he can compute $f(c)$ by a mod $q-1$ operation, and then from $f(c)$ the user can compute $c$ by inverting the isomorphic mapping.

To account for potential zero values of ${\bf A}_s, {\bf B}_s\in\mathbb{F}_q$, let us define $f(0)=0$. 
In light of this discussion, let us assume $f({\mathbf{A}}_{s}), f({\mathbf{B}}_{s})$ are scalars in $\mathbb{F}_p$ and the user wishes to retrieve $f({\mathbf{A}}_s)+f({\mathbf{B}}_s)\in\mathbb{F}_p$ for all those $s\in[S]$ where ${\bf A}_s, {\bf B}_s$ are both non-zero, and he wishes to retrieve the answer $0$  for all those $s\in[S]$ where either one of ${\bf A}_s, {\bf B}_s$  is zero. Now let us present a scheme to achieve this task. For this scheme let us choose $p$  to be the minimum prime number such that $p>2(q-1)$. Let  $\mathbf{Z}_{sx}, \mathbf{Z}'_{sx}$, $s\in[S], x\in[X]$ be uniformly distributed random (noise) scalars over $\mathbb{F}_p$. The independence of these random scalars and the scalars $f(\mathbf{A}_{s}), f(\mathbf{B}_{s})$ is specified as follows.
\begin{align}
&H_q\left((\mathbf{Z}_{sx},\mathbf{Z}'_{sx})_{s\in[S],x\in[X]},(f(\mathbf{A}_{s}), f(\mathbf{B}_{s}))_{s\in[S]}\right)\notag\\
=&\sum_{s\in[S],x\in[X]}H_q(\mathbf{Z}_{sx})+\sum_{s\in[S],x\in[X]}H_q(\mathbf{Z}'_{sx})+\sum_{s\in[S]}H_q(f(\mathbf{A}_{s}))+\sum_{s\in[S]}H_q(f(\mathbf{B}_{s})).
\end{align}

Let $\alpha_n, n\in[N]$ be $N$ distinct  elements from $\mathbb{F}_p$. The construction of $\widetilde{A}_s^n$ and $\widetilde{B}_s^n$ for any $s\in[S]$ and $n\in[N]$ is provided as follows.
\begin{align}
\widetilde{A}_s^n&=\alpha_n^sf(\mathbf{A}_s)+\sum_{x\in[X]}\alpha_n^{S+x}\mathbf{Z}_{sx}\\
&=\alpha_n^sf(\mathbf{A}_s)+\alpha_n^{S+1}\mathbf{Z}_{s1}+\dots+\alpha_n^{S+X}\mathbf{Z}_{sX}\\
\widetilde{B}_s^n&=\alpha_n^sf(\mathbf{B}_s)+\sum_{x\in[X]}\alpha_n^{S+x}\mathbf{Z}'_{sx}\\
&=\alpha_n^sf(\mathbf{B}_s)+\alpha_n^{S+1}\mathbf{Z}'_{s1}+\dots+\alpha_n^{S+X}\mathbf{Z}'_{sX}.
\end{align}
The answer from the $n$-th server is obtained as follows
\begin{align}
\Delta_n&=\widetilde{A}_1^n+\dots+\widetilde{A}_S^n+\widetilde{B}_1^n+\dots+\widetilde{B}_S^n
\\
&=
\alpha_n(f(\mathbf{A}_1)+f(\mathbf{B}_1))+\dots+\alpha_n^S(f(\mathbf{A}_S)+f(\mathbf{B}_S))\notag\\
&\quad+\alpha_n^{S+1}\sum_{s\in[S]}(\mathbf{Z}_{s1}+\mathbf{Z}'_{s1})+\dots+\alpha_n^{S+X}\sum_{s\in[S]}(\mathbf{Z}_{sX}+\mathbf{Z}'_{sX}).
\end{align}
$X$-security is guaranteed because matrices $\mathbf{A}_s, \mathbf{B}_s$ are protected by  noise terms that are i.i.d. uniform and coded according to $\text{MDS}(X,N)$ codes. Note that desired scalars and random scalars are coded with an RS code, and $S+X=N$. Therefore, from the answers provided by all $N$ servers, the user is able to decode $(\mathbf{A}_s+\mathbf{B}_s)_{s\in[S]}$. By the transformation argument, correctness is guaranteed for all those $s\in[S]$, where  $\mathbf{A}_s\neq 0$ and $\mathbf{B}_s\neq 0$. However, correctness is not yet guaranteed for those $s\in[S]$ where either $\mathbf{A}_s= 0$ or $\mathbf{B}_s= 0$. For this we will implement a separate mechanism to let the user know which $\mathbf{A}_s$ and $\mathbf{B}_s$ are equal to zero, so he can infer correctly that ${\bf A}_s{\bf B}_s=0$ for those instances. Specifically, for each scalar  $\mathbf{A}_{s}$ and $\mathbf{B}_{s}$, let us define  binary symbols $\eta_{{\bf A}_s}, \eta_{{\bf B}_s}$ that indicate whether or not $\mathbf{A}_{s}, \mathbf{B}_{s}$ are equal to zero. These $\eta_{{\bf A}_s}, \eta_{{\bf B}_s}$ are  also secret-shared among the $N$ servers in an $X$-secure fashion, and retrieved by the user at negligible increase in download cost  as $q\rightarrow\infty$.  Now, by the Bertrand-Chebyshev theorem, for every integer $\nu>1$ there is always at least one prime $p'$ such that $\nu<p'<2\nu$, thus we must have $p$ such that $2(q-1)<p<4(q-1)$. Therefore, as $q\rightarrow\infty$, the rate achieved is
\begin{align}
R&=\frac{H_q\left(\mathbf{A}\times\mathbf{B}\right)}{D}\\
&\geq  \frac{H_q\left({\bf AB}\mid {\bf B}\right)}{N\log_q(p)+2SN\log_q(2)}\\
&\geq\frac{S (\frac{q-1}{q})}{N\log_q(4(q-1))+2SN\log_q(2)}\\
&=\frac{(N-X)(\frac{q-1}{q})}{N\log_q(4(q-1))+2(N-X)N\log_q(2)}
\end{align}
which approaches $1-X/N$ as $q\rightarrow\infty$.
\subsubsection{$K=1$, Arbitrary $L,M$}
Now let us consider the setting with arbitrary $L,M$, and with $K=1$, i.e., the user wishes to compute the outer product of vectors ${\bf A}_s, {\bf B}_s$ for all $s\in[S]$. As before, let us set $S=N-X$, and choose $p$ to be the smallest prime such that $p>2(q-1)$. We will allow the user to download a normalized version of each ${\bf A}_s$ and ${\bf B}_s$ vector, along with the product of the normalizing factors, from which the user can construct ${\bf A}_s{\bf B}_s$. To this end, let us define $i_s, j_s$ as the index of the first non-zero element in ${\bf A}_s, {\bf B}_s$, respectively, and normalize 
 each vector $\mathbf{A}_{s}$ by it's $i_s^{th}$ element ${A}_s(i_s)$, each vector $\mathbf{B}_{s}$ by it's $j_s^{th}$ element ${B}_s(j_s)$ $\forall s\in[S]$.  Now $\forall s\in[S]$ such that $\mathbf{A}_s$ is not the zero vector, we have
\begin{align}
\mathbf{A}_s=A_s(i_s)\underbrace{[0,\cdots,0,1,A'_s(i_s+1),\dots,A'_s(L)]'}_{\mathbf{A}'_s},
\end{align}
where the vector ${\bf A}_s'$ is the normalized vector. Similarly, $\forall s\in[S]$ such that $\mathbf{B}_s$ is not the zero vector, we have
\begin{align}
\mathbf{B}_s=B_s(j_s)\underbrace{[0,\cdots, 0,1,B'_s(j_s+1),\dots,B'_s(M)]}_{\mathbf{B}'_s}.
\end{align}
Note that if ${\bf A}_s$ or ${\bf B}_s$ is the zero vector, then we simply set $i_s=0, j_s=0$ and ${\bf A}_s'=0, {\bf B}_s'=0$,  and $A_s(i_s)=1, B_s(j_s)=1$, respectively.
The scheme is constructed as follows. Separate $X$-secure secret sharing schemes are used to distribute the secrets $i_s, j_s, A_s(i_s), B_s(j_s), \bar{\bf A}_s', \bar{\bf B}_s'$, among the $N$ servers, where $\bar{\bf A}_s', \bar{\bf B}_s'$ are length $L-1, M-1$ vectors respectively, obtained by eliminating the leading $1$ term from each of ${\bf A}_s', {\bf B}_s'$ (or one of the zeros if ${\bf A}_s', {\bf B}_s'$ are zero vectors).  The vectors $\bar{\bf A}_s', \bar{\bf B}_s'$ are  retrieved  by the user according to the scheme presented in Section \ref{sec:generalscheme}, with total download cost equal to $N(L-1)+N(M-1)=N(L+M-2)$ $q$-ary symbols. The  indices $i_s, j_s$ are retrieved according to same scheme presented in Section \ref{sec:generalscheme}, with total download cost not exceeding $N\log_q(L+1)+N\log_q(M+1)$ in units of $q$-ary symbols because  the alphabet size for the indices $i_s, j_s$ is $L+1, M+1$, respectively. 
The scheme presented in Section \ref{sec:achfs11} is utilized by the user to retrieve the scalar products $A_s(i_s)B_s(j_s)$ with total download from $N$ servers not exceeding $N\log_q(4(q-1))$.  Note that since $A_s(i_s)B_s(j_s)$ is always non-zero there is no need for  downloading additional indicators needed to identify zero values. The correctness follows from the fact that
\begin{align}
\mathbf{A}_s\times\mathbf{B}_s=(A_s(i_s)B_s(j_s))\mathbf{A}'_s\times\mathbf{B}'_s,
\end{align}
and by the construction of the scheme, $(A_s(i_s)B_s(j_s))$ and $\mathbf{A}'_s\times\mathbf{B}'_s$ are recoverable for all $s\in[S]$.
Therefore,  the rate achieved is
\begin{align}
R&=\frac{H_q\left((\mathbf{A}_s\times\mathbf{B}_s)_{s\in[S]}\right)}{D}\\
&\geq\frac{H_q\left((\mathbf{A}_s\times\mathbf{B}_s)_{s\in[S]}\right)}{N(L+M-2)+N\log_q(L+1)+N\log_q(M+1)+N\log_q(4(q-1))}\\
&=\frac{S(L+M-1)}{N(L+M-2)+N\log_q(L+1)+N\log_q(M+1)+N\log_q(4(q-1))}\\
&=\frac{(N-X)(L+M-1)}{N(L+M-2)+N\log_q(L+1)+N\log_q(M+1)+N\log_q(4(q-1))}
\end{align}
which approaches $1-\frac{X}{N}$ as $q\rightarrow\infty$.
This proves achievability of case  \eqref{eq:cabphi2}, and completes the achievability proof of  Theorem \ref{thm:cabphi}.

\section{Conclusion}
A class of Secure Distributed Batch Matrix Multiplication (SDBMM) problems was defined in this work, and its capacity characterized in various parameter regimes depending on the security level $X$, number of servers $N$, matrix dimensions $L, M, K$ and the set of matrices that are secured or available to the users as side-information. Notable aspects include connections between SDBMM and a form of PIR known as MM-XSTPIR that led us to various converse bounds, and cross-subspace alignment schemes along with monomorphic transformations from scalar multiplication to scalar addition that formed the basis of some of the achievable schemes. 
Note that  most of the achievable schemes in this work can also be adapted to the one-shot matrix multiplication framework of \cite{Kakar_Ebadifar_Sezgin_CSA, Chang_Tandon_SDMMOS}, say where $L, K, M$ all approach infinity and the ratios $L/K$, $K/M$ are fixed constants. The converse parts follow directly, for the achievability parts, we can adapt our schemes to one-shot matrix multiplication based on matrix partitioning and zero-padding. For example, consider the \emph{cross subspace alignment} based scheme in Section \ref{sec:csascheme}. Let us define $L_1=\lfloor L/S'\rfloor S'$, where $S'=N-X_A-X_B$. Now let us partition matrix $\mathbf{A}$ as follows.
\begin{equation}
\mathbf{A}=[\mathbf{A}_{1,1} \quad \dots \quad \mathbf{A}_{1,S'} \quad \mathbf{A}_{2,1} \quad \dots \quad \mathbf{A}_{2,(L \mod S')}]^T,
\end{equation} 
where $(\mathbf{A}_{1,s'})_{s\in[S']}$ are $L_1/S'\times K$ matrices, $(\mathbf{A}_{2,s'})_{s\in[L\mod S']}$ are $1\times K$ vectors. Now with the schemes presented in this work, for all $s'\in[S']$, by setting $\mathbf{A}_{s'}=\mathbf{A}_{1,s'}$ and $\mathbf{B}_{s'}=\mathbf{B}$, one can recover $(\mathbf{A}_{1,s'}\mathbf{B})_{s\in[S']}$. On the other hand, we can then set $\mathbf{A}_{s'}=\mathbf{A}_{2,s'}, s'\in[L\mod S']$, $\mathbf{A}_{s'}=\mathbf{0}, s\in\{(L\mod S')+1,\dots,S'\}$ and also $\mathbf{B}_{s'}=\mathbf{B}, s\in[S']$ to recover $(\mathbf{A}_{2,s'}\mathbf{B})_{s\in[K\mod S']}$. Note that the extra cost of zero-padding is upper bounded by $KS'/LK=S'/L$, which goes to zero as $L\rightarrow\infty$. Thus the desired rates are still achievable. In terms of future work, open problems that merit immediate attention include the many cases of SDBMM where the capacity remains open. For example,  the capacity of the basic SDBMM$_{(\bf AB,\phi)}$ setting, previously believed to be solved in \cite{Kakar_Ebadifar_Sezgin_CSA} is shown to be still open in general, including the  important case of square matrices $L=K=M>1$ with sufficiently many servers $N>X$. From the case $L=K=M=1$ that is already solved in this work, it seems that the generalization could require expanding the scope of constructions based on non-trivial monomorphic transformations, which presents an interesting research avenue. In terms of the connection to PIR, this work  highlights the importance of finding the capacity characterizations for MM-XSTPIR, as well as MM-XSTPC. Evidently solutions to these PIR problems would not only add to the growing literature on PIR that already includes many successful capacity characterizations \cite{Sun_Jafar_PIR, Sun_Jafar_TPIR, Sun_Jafar_SPIR, Banawan_Ulukus, Jia_Sun_Jafar_XSTPIR, Sun_Jafar_PC, Mirmohseni_Maddah}, but also have a ripple effect on important problems  that are intimately connected to PIR. Similar to PIR, the models of SDBMM could also be further enriched to include privacy of retrieved information, coded storage \cite{Banawan_Ulukus, FREIJ_HOLLANTI, Sun_Jafar_MDSTPIR}, storage size and repair constraints \cite{Tandon_Amuru_Clancy_Buehrer, Attia_Kumar_Tandon, Wei_Banawan_Ulukus} and generalized forms of side-information \cite{Kadhe_Garcia_Heidarzadeh_Rouayheb_Sprintson, Tandon_CPIR, Chen_Wang_Jafar}. Thus, just like PIR, SDBMM offers a fertile research landscape for discovering new coding structures and  converse arguments.

\appendix
\section{Upper Bound of the Capacity of Multi-Message XSTPIR}\label{sec:proofmmxstpir}
To prove Theorem \ref{thm:mmxstpir}, we need following lemmas.
\begin{lemma}\label{lemma:mmxstpir1}
For all $\mathcal{K},\mathcal{K}'\subset[K]$, $|\mathcal{K}|=|\mathcal{K}'|=M$, $\forall\mathcal{T}\subset[N], |\mathcal{T}|=T$, we have
\begin{align}
\left(Q_{\mathcal{T}}^{\mathcal{K}},A_{\mathcal{T}}^{\mathcal{K}},S_{[N]},W_{[K]}\right)\sim\left(Q_{\mathcal{T}}^{\mathcal{K}'},A_{\mathcal{T}}^{\mathcal{K}'},S_{[N]},W_{[K]}\right)
\end{align}
\end{lemma}
\begin{proof}
It suffices to prove $I_q\left(Q_{\mathcal{T}}^{\mathcal{K}},A_{\mathcal{T}}^{\mathcal{K}},S_{[N]},W_{[K]};\mathcal{K}\right)=0$. The proof is presented as follows.
\begin{align}
&I_q\left(Q_{\mathcal{T}}^{\mathcal{K}},A_{\mathcal{T}}^{\mathcal{K}},S_{[N]},W_{[K]};\mathcal{K}\right)\\
=&I_q(Q_{\mathcal{T}}^{\mathcal{K}};\mathcal{K})+I_q(S_{[N]},W_{[K]};\mathcal{K}|Q_{\mathcal{T}}^{\mathcal{K}})+I_q(A_{\mathcal{T}}^{\mathcal{K}};\mathcal{K}|Q_{\mathcal{T}}^{\mathcal{K}},S_{[N]},W_{[K]})\label{eq:lemmammxstpir11}\\
=&I_q(Q_{\mathcal{T}}^{\mathcal{K}};\mathcal{K})+I_q(S_{[N]},W_{[K]};\mathcal{K}|Q_{\mathcal{T}}^{\mathcal{K}})\label{eq:lemmammxstpir12}\\
=&I_q(Q_{\mathcal{T}}^{\mathcal{K}};\mathcal{K})+I_q(S_{[N]};\mathcal{K}|Q_{\mathcal{T}}^{\mathcal{K}})\label{eq:lemmammxstpir13}\\
\leq&I_q(Q_{\mathcal{T}}^{\mathcal{K}};\mathcal{K})+I_q(S_{[N]};\mathcal{K},Q_{\mathcal{T}}^{\mathcal{K}})\label{eq:lemmammxstpir14}\\
=&0.\label{eq:lemmammxstpir15}
\end{align}
Steps are justified as follows. \eqref{eq:lemmammxstpir11} is the chain rule of mutual information. \eqref{eq:lemmammxstpir12} holds from the fact that $A_{\mathcal{T}}^{\mathcal{K}}$ is function of $(Q_{\mathcal{T}}^{\mathcal{K}},S_{[N]})$ according to \eqref{eq:mmxstpiransfunc}. \eqref{eq:lemmammxstpir13} follows because $W_{[K]}$ is function of $S_{[N]}$, according to \eqref{eq:mmxstpirmsgfunc}. \eqref{eq:lemmammxstpir14} follows from the chain rule and non-negativity of mutual information. In last step, we simply used \eqref{eq:mmxstpirqsind} and \eqref{eq:mmxstpirpriv}. This completes the proof of Lemma \ref{lemma:mmxstpir1}.
\end{proof}

\begin{lemma}\label{lemma:mmxstpir2}
For all $\mathcal{T},\mathcal{X}\subset[N]$, $|\mathcal{T}|=T, |\mathcal{X}|=X$, $\forall\mathcal{K},\mathcal{K}'\subset[K]$, $|\mathcal{K}|=|\mathcal{K}'|=M$, $\forall\kappa\subset[K]$, we have
\begin{align}
H_q(A_{\mathcal{T}}^{\mathcal{K}}|S_{\mathcal{X}},Q_{[N]}^{\mathcal{K}},W_{\kappa})=H_q(A_{\mathcal{T}}^{\mathcal{K}}|S_{\mathcal{X}},Q_{\mathcal{T}}^{\mathcal{K}},W_{\kappa}).
\end{align}
\end{lemma}
\begin{proof}
\begin{align}
&H_q(A_{\mathcal{T}}^{\mathcal{K}}|S_{\mathcal{X}},Q_{\mathcal{T}}^{\mathcal{K}},W_{\kappa})-H_q(A_{\mathcal{T}}^{\mathcal{K}}|S_{\mathcal{X}},Q_{[N]}^{\mathcal{K}},W_{\kappa})\\
&=I_q(A_{\mathcal{T}}^{\mathcal{K}};Q_{[N]}^{\mathcal{K}}|S_{\mathcal{X}},Q_{\mathcal{T}}^{\mathcal{K}},W_{\kappa})\label{eq:lemmammxstpir21}\\
&\leq I_q(A_{\mathcal{T}}^{\mathcal{K}},S_{\mathcal{X}},W_{\kappa};Q_{[N]}^{\mathcal{K}}|Q_{\mathcal{T}}^{\mathcal{K}})\label{eq:lemmammxstpir22}\\
&\leq I_q(A_{\mathcal{T}}^{\mathcal{K}},S_{[N]},W_{\kappa};Q_{[N]}^{\mathcal{K}}|Q_{\mathcal{T}}^{\mathcal{K}})\label{eq:lemmammxstpir23}\\
&=I_q(A_{\mathcal{T}}^{\mathcal{K}},S_{[N]};Q_{[N]}^{\mathcal{K}}|Q_{\mathcal{T}}^{\mathcal{K}})\label{eq:lemmammxstpir24}\\
&=I_q(S_{[N]};Q_{[N]}^{\mathcal{K}}|Q_{\mathcal{T}}^{\mathcal{K}})+I_q(A_{\mathcal{T}}^{\mathcal{K}};Q_{[N]}^{\mathcal{K}}|Q_{\mathcal{T}}^{\mathcal{K}},S_{[N]})\label{eq:lemmammxstpir25}\\
&=I_q(S_{[N]};Q_{[N]}^{\mathcal{K}}|Q_{\mathcal{T}}^{\mathcal{K}})\label{eq:lemmammxstpir26}\\
&\leq I_q(S_{[N]};Q_{[N]}^{\mathcal{K}})\label{eq:lemmammxstpir27}\\
&=0.\label{eq:lemmammxstpir28}
\end{align}
Steps are justified as follows. \eqref{eq:lemmammxstpir21} is the definition of mutual information. \eqref{eq:lemmammxstpir22} follows from the chain rule and non-negativity of mutual information. In \eqref{eq:lemmammxstpir23}, we added terms in mutual information. \eqref{eq:lemmammxstpir24} holds from the fact that $W_{[K]}$ is function of $S_{[N]}$, according to \eqref{eq:mmxstpirmsgfunc}. \eqref{eq:lemmammxstpir25} is the chain rule of mutual information. \eqref{eq:lemmammxstpir26} follows from the fact that $A_{\mathcal{T}}^{\mathcal{K}}$ is fully determined by $(Q_{\mathcal{T}}^{\mathcal{K}},S_{[N]})$ according to \eqref{eq:mmxstpiransfunc}. \eqref{eq:lemmammxstpir27} follows from the chain rule and non-negativity of mutual information, while the last step holds from \eqref{eq:mmxstpirqsind}. This completes the proof of Lemma \ref{lemma:mmxstpir2}.
\end{proof}

\begin{lemma}\label{lemma:mmxstpir3}
Denote $D_n$ the expected number of $q$-ary symbols downloaded from the $n$-th server. For all $\mathcal{X}\subset[N]$, $|\mathcal{X}|=X$, $\overline{\mathcal{X}}=[N]\setminus\mathcal{X}$, $\forall\mathcal{K}_1\subset[K]$, $|\mathcal{K}_1|=M$, we have
\begin{align}
ML\leq\sum_{n\in\overline{\mathcal{X}}}D_n-H_q(A_{\overline{\mathcal{X}}}^{\mathcal{K}_1}|S_{\mathcal{X}},Q_{[N]}^{\mathcal{K}_1},W_{\mathcal{K}_1}).
\end{align}
\end{lemma}
\begin{proof}
\begin{align}
ML=H_q(W_{\mathcal{K}_1})&=I_q(W_{\mathcal{K}_1};A_{[N]}^{\mathcal{K}_1}|Q_{[N]}^{\mathcal{K}_1})\label{eq:lemmammxstpir31}\\
&\leq I_q(W_{\mathcal{K}_1};A_{[N]}^{\mathcal{K}_1},S_{\mathcal{X}}|Q_{[N]}^{\mathcal{K}_1})\label{eq:lemmammxstpir32}\\
&=I_q(W_{\mathcal{K}_1};S_{\mathcal{X}}|Q_{[N]}^{\mathcal{K}_1})+I_q(W_{\mathcal{K}_1};A_{[N]}^{\mathcal{K}_1}|S_{\mathcal{X}},Q_{[N]}^{\mathcal{K}_1})\label{eq:lemmammxstpir33}\\
&\leq I_q(W_{\mathcal{K}_1},Q_{[N]}^{\mathcal{K}_1};S_{\mathcal{X}})+I_q(W_{\mathcal{K}_1};A_{[N]}^{\mathcal{K}_1}|S_{\mathcal{X}},Q_{[N]}^{\mathcal{K}_1})\label{eq:lemmammxstpir34}\\
&=I_q(W_{\mathcal{K}_1};S_{\mathcal{X}})+I_q(Q_{[N]}^{\mathcal{K}_1};S_{\mathcal{X}}\mid W_{\mathcal{K}_1})+I_q(W_{\mathcal{K}_1};A_{[N]}^{\mathcal{K}_1}|S_{\mathcal{X}},Q_{[N]}^{\mathcal{K}_1})\label{eq:lemmammxstpir35}\\
&\leq I_q(Q_{[N]}^{\mathcal{K}_1};W_{\mathcal{K}_1},S_{\mathcal{X}})+I_q(W_{\mathcal{K}_1};A_{[N]}^{\mathcal{K}_1}|S_{\mathcal{X}},Q_{[N]}^{\mathcal{K}_1})\label{eq:lemmammxstpir36}\\
&=I_q(W_{\mathcal{K}_1};A_{\mathcal{X}}^{\mathcal{K}_1},A_{\overline{\mathcal{X}}}^{\mathcal{K}_1}|S_{\mathcal{X}},Q_{[N]}^{\mathcal{K}_1})\label{eq:lemmammxstpir37}\\
&=I_q(W_{\mathcal{K}_1};A_{\overline{\mathcal{X}}}^{\mathcal{K}_1}|S_{\mathcal{X}},Q_{[N]}^{\mathcal{K}_1})\label{eq:lemmammxstpir38}\\
&=H_q(A_{\overline{\mathcal{X}}}^{\mathcal{K}_1}|S_{\mathcal{X}},Q_{[N]}^{\mathcal{K}_1})-H_q(A_{\overline{\mathcal{X}}}^{\mathcal{K}_1}|W_{\mathcal{K}_1},S_{\mathcal{X}},Q_{[N]}^{\mathcal{K}_1})\label{eq:lemmammxstpir39}\\
&\leq\sum_{n\in\overline{\mathcal{X}}}D_n-H_q(A_{\overline{\mathcal{X}}}^{\mathcal{K}_1}|S_{\mathcal{X}},Q_{[N]}^{\mathcal{K}_1},W_{\mathcal{K}_1}).\label{eq:lemmammxstpir310}
\end{align}
Steps are justified as follows. \eqref{eq:lemmammxstpir31} follows from \eqref{eq:mmxstpirdecode}, while in \eqref{eq:lemmammxstpir32}, we add terms in mutual information. In \eqref{eq:lemmammxstpir33}, \eqref{eq:lemmammxstpir34}, \eqref{eq:lemmammxstpir35} and \eqref{eq:lemmammxstpir36}, we repeatedly used  the chain rule and non-negativity of mutual information, while \eqref{eq:lemmammxstpir36} and \eqref{eq:lemmammxstpir37} holds from the independence of query and storage, according to \eqref{eq:mmxstpirqsind}. \eqref{eq:lemmammxstpir38} holds from the fact that $A_{\mathcal{X}}^{\mathcal{K}_1}$ is fully determined by $(Q_{[N]}^{\mathcal{K}_1},S_{\mathcal{X}})$ according to \eqref{eq:mmxstpiransfunc}. \eqref{eq:lemmammxstpir39}
 is the definition of mutual information, while \eqref{eq:lemmammxstpir310} follows from the fact that dropping conditions can not reduce entropy. This completes the proof of Lemma \ref{lemma:mmxstpir3}.
\end{proof}

\begin{lemma}\label{lemma:mmxstpir4}
For all $\mathcal{X}\subset[N]$, $|\mathcal{X}|=X$, $\forall\mathcal{K}\subset[K]$, $|\mathcal{K}|=M$, $\forall \kappa\subset[K]$, we have
\begin{align}
I_q(W_{\mathcal{K}};S_{\mathcal{X}},Q_{[N]}^{\mathcal{K}},W_{\kappa})=|\mathcal{K}\cap\kappa|L.
\end{align}
\end{lemma}
\begin{proof}
\begin{align}
&I_q(W_{\mathcal{K}};S_{\mathcal{X}},Q_{[N]}^{\mathcal{K}},W_{\kappa})\notag\\
&=I_q(W_{\mathcal{K}};W_{\kappa})+I_q(W_{\mathcal{K}};S_{\mathcal{X}},Q_{[N]}^{\mathcal{K}}|W_{\kappa})\label{eq:lemmammxstpir41}\\
&=|\mathcal{K}\cap\kappa|L+I_q(W_{\mathcal{K}};S_{\mathcal{X}},Q_{[N]}^{\mathcal{K}}|W_{\kappa}).\label{eq:lemmammxstpir42}
\end{align}
\eqref{eq:lemmammxstpir41} is the chain rule of mutual information, and \eqref{eq:lemmammxstpir42} follows from \eqref{eq:mmxstpirmsgentropy} and \eqref{eq:mmxstpirmsgind}. Let us consider the RHS term, we have
\begin{align}
&I_q(W_{\mathcal{K}};S_{\mathcal{X}},Q_{[N]}^{\mathcal{K}}|W_{\kappa})\notag\\
&\leq I_q(S_{\mathcal{X}},Q_{[N]}^{\mathcal{K}};W_{\mathcal{K}},W_{\kappa})\label{eq:lemmammxstpir43}\\
&=I_q(S_{\mathcal{X}};W_{\mathcal{K}},W_{\kappa})+I_q(Q_{[N]}^{\mathcal{K}};W_{\mathcal{K}},W_{\kappa}|S_{\mathcal{X}})\label{eq:lemmammxstpir44}\\
&=I_q(Q_{[N]}^{\mathcal{K}};W_{\mathcal{K}},W_{\kappa}|S_{\mathcal{X}})\label{eq:lemmammxstpir45}\\
&\leq I_q(Q_{[N]}^{\mathcal{K}};W_{\mathcal{K}},W_{\kappa},S_{\mathcal{X}})\label{eq:lemmammxstpir46}\\
&\leq I_q(Q_{[N]}^{\mathcal{K}};S_{[N]})\label{eq:lemmammxstpir47}\\
&=0.\label{eq:lemmammxstpir48}
\end{align}
Steps are justified as follows. \eqref{eq:lemmammxstpir43} and \eqref{eq:lemmammxstpir44} follows from the chain rule and non-negativity of mutual information, while \eqref{eq:lemmammxstpir45} follows from the $X$-secure constraint in \eqref{eq:mmxstpirsecur}. \eqref{eq:lemmammxstpir46} holds from the chain rule and non-negativity of mutual information, and \eqref{eq:lemmammxstpir47} follows from the fact that $(W_{\mathcal{K}},W_{\kappa},S_{\mathcal{X}})$ is function of $S_{[N]}$. The last step follows from \eqref{eq:mmxstpirqsind}. This completes the proof of Lemma \ref{lemma:mmxstpir4}.
\end{proof}
Now we are ready to formally present the proof of Theorem \ref{thm:mmxstpir}.
\begin{proof}
First, let us consider $X<N\leq X+T$. For this setting, let us assume that $\mathcal{K}_i=[i:i+M-1], i\in[K-M+1]$. Note that by the selection of $\mathcal{K}_i$'s, $\forall i\in[K-M]$, we have
\begin{align}\label{eq:mmxstpirk1}
|\mathcal{K}_{i+1}\cap(\mathcal{K}_1\cup\cdots\cup\mathcal{K}_{i})|=(M-1).
\end{align}
Now let us consider the RHS term in \eqref{eq:lemmammxstpir310}. For all $i\in[K-M]$, we have
\begin{align}
&H_q(A_{\overline{\mathcal{X}}}^{\mathcal{K}_i}|S_{\mathcal{X}},Q_{[N]}^{\mathcal{K}_i},W_{\mathcal{K}_1\cup\cdots\cup\mathcal{K}_{i}})\notag\\
&=H_q(A_{\overline{\mathcal{X}}}^{\mathcal{K}_i}|S_{\mathcal{X}},Q_{\overline{\mathcal{X}}}^{\mathcal{K}_i},W_{\mathcal{K}_1\cup\cdots\cup\mathcal{K}_{i}})\label{eq:thmmmxstpir11}\\
&=H_q(A_{\overline{\mathcal{X}}}^{\mathcal{K}_{i+1}}|S_{\mathcal{X}},Q_{\overline{\mathcal{X}}}^{\mathcal{K}_{i+1}},W_{\mathcal{K}_1\cup\cdots\cup\mathcal{K}_{i}})\label{eq:thmmmxstpir12}\\
&=H_q(A_{\overline{\mathcal{X}}}^{\mathcal{K}_{i+1}}|S_{\mathcal{X}},Q_{[N]}^{\mathcal{K}_{i+1}},W_{\mathcal{K}_1\cup\cdots\cup\mathcal{K}_{i}})\label{eq:thmmmxstpir13}\\
&=H_q(W_{\mathcal{K}_{i+1}},A_{\overline{\mathcal{X}}}^{\mathcal{K}_{i+1}}|S_{\mathcal{X}},Q_{[N]}^{\mathcal{K}_{i+1}},W_{\mathcal{K}_1\cup\cdots\cup\mathcal{K}_{i}})\label{eq:thmmmxstpir14}\\
&=H_q(W_{\mathcal{K}_{i+1}}|S_{\mathcal{X}},Q_{[N]}^{\mathcal{K}_{i+1}},W_{\mathcal{K}_1\cup\cdots\cup\mathcal{K}_{i}})+H_q(A_{\overline{\mathcal{X}}}^{\mathcal{K}_{i+1}}|S_{\mathcal{X}},Q_{[N]}^{\mathcal{K}_{i+1}},W_{\mathcal{K}_1\cup\cdots\cup\mathcal{K}_{i+1}})\label{eq:thmmmxstpir15}\\
&=L+H_q(A_{\overline{\mathcal{X}}}^{\mathcal{K}_{i+1}}|S_{\mathcal{X}},Q_{[N]}^{\mathcal{K}_{i+1}},W_{\mathcal{K}_1\cup\cdots\cup\mathcal{K}_{i+1}}).\label{eq:thmmmxstpir16}
\end{align}
Steps are justified as follows. \eqref{eq:thmmmxstpir11} follows from Lemma \ref{lemma:mmxstpir2}, while \eqref{eq:thmmmxstpir12} follows from Lemma \ref{lemma:mmxstpir1}. \eqref{eq:thmmmxstpir13} again follows from Lemma \ref{lemma:mmxstpir2}. \eqref{eq:thmmmxstpir14} follows from \eqref{eq:mmxstpiransfunc} and \eqref{eq:mmxstpirdecode}. \eqref{eq:thmmmxstpir15} is the chain rule of entropy, while the last step follows from Lemma \ref{lemma:mmxstpir4} and \eqref{eq:mmxstpirk1}. Applying \eqref{eq:thmmmxstpir16} repeatedly for $i=1,2,\dots,K-M$, we have
\begin{align}
ML&\leq\sum_{n\in\overline{\mathcal{X}}}D_n-H_q(A_{\overline{\mathcal{X}}}^{\mathcal{K}_1}|S_{\mathcal{X}},Q_{[N]}^{\mathcal{K}_1},W_{\mathcal{K}_1})\\
&=\sum_{n\in\overline{\mathcal{X}}}D_n-L-H_q(A_{\overline{\mathcal{X}}}^{\mathcal{K}_2}|S_{\mathcal{X}},Q_{[N]}^{\mathcal{K}_2},W_{\mathcal{K}_1\cup\mathcal{K}_2})\\
&=\cdots\\
&=\sum_{n\in\overline{\mathcal{X}}}D_n-(K-M)L.
\end{align}
Averaging over all $\mathcal{X}$, we have
\begin{align}
D=\sum_{n\in[N]}D_n\geq\frac{N}{N-X}KL.
\end{align}
Therefore we have
\begin{align}
R=\frac{ML}{D}\leq\frac{M(N-X)}{KN}.
\end{align}
Thus
\begin{align}
C_{\text{MM-XSTPIR}}(N,X,T,K,M)\leq\frac{M(N-X)}{KN}, \quad X<N\leq X+T.
\end{align}
Next, let us consider $N>X+T$. For this setting, let us assume that $\mathcal{K}_i=\{M(i-1)+1,M(i-1)+2,\dots,Mi\}$, $\forall i\in[\lfloor\frac{K}{M}\rfloor]$. Note that $\mathcal{K}_i$'s are disjoint sets. Similarly, let us consider the RHS term in \eqref{eq:lemmammxstpir310}. Consider any set $\mathcal{T}\subset\overline{\mathcal{X}}$, $|\mathcal{T}|=T$, For all $i,i+1\in[\lfloor\frac{K}{M}\rfloor]$, we have
\begin{align}
&H_q(A_{\overline{\mathcal{X}}}^{\mathcal{K}_i}|S_{\mathcal{X}},Q_{[N]}^{\mathcal{K}_i},W_{\mathcal{K}_1\cup\cdots\cup\mathcal{K}_{i}})\notag\\
&\geq H_q(A_{\mathcal{T}}^{\mathcal{K}_i}|S_{\mathcal{X}},Q_{[N]}^{\mathcal{K}_i},W_{\mathcal{K}_1\cup\cdots\cup\mathcal{K}_{i}})\label{eq:thmmmxstpir21}\\
&=H_q(A_{\mathcal{T}}^{\mathcal{K}_i}|S_{\mathcal{X}},Q_{\mathcal{T}}^{\mathcal{K}_i},W_{\mathcal{K}_1\cup\cdots\cup\mathcal{K}_{i}})\label{eq:thmmmxstpir22}\\
&=H_q(A_{\mathcal{T}}^{\mathcal{K}_{i+1}}|S_{\mathcal{X}},Q_{\mathcal{T}}^{\mathcal{K}_{i+1}},W_{\mathcal{K}_1\cup\cdots\cup\mathcal{K}_{i}})\label{eq:thmmmxstpir23}\\
&=H_q(A_{\mathcal{T}}^{\mathcal{K}_{i+1}}|S_{\mathcal{X}},Q_{[N]}^{\mathcal{K}_{i+1}},W_{\mathcal{K}_1\cup\cdots\cup\mathcal{K}_{i}}).\label{eq:thmmmxstpir24}
\end{align}
Steps are justified as follows. \eqref{eq:thmmmxstpir21} follows from the fact that dropping terms can not increase entropy. \eqref{eq:thmmmxstpir22} follows from Lemma \ref{lemma:mmxstpir2}. \eqref{eq:thmmmxstpir23} follows from Lemma \ref{lemma:mmxstpir1}, while \eqref{eq:thmmmxstpir24} again follows from Lemma \ref{lemma:mmxstpir2}. Now let us average \eqref{eq:thmmmxstpir24} over all $\mathcal{T}$ and apply Han's inequality. 
\begin{align}
&H_q(A_{\overline{\mathcal{X}}}^{\mathcal{K}_i}|S_{\mathcal{X}},Q_{[N]}^{\mathcal{K}_i},W_{\mathcal{K}_1\cup\cdots\cup\mathcal{K}_{i}})\notag\\
&\geq \frac{T}{N-X}H_q(A_{\overline{\mathcal{X}}}^{\mathcal{K}_{i+1}}|S_{\mathcal{X}},Q_{[N]}^{\mathcal{K}_{i+1}},W_{\mathcal{K}_1\cup\cdots\cup\mathcal{K}_{i}})\label{eq:thmmmxstpir25}\\
&=\frac{T}{N-X}H_q(W_{\mathcal{K}_{i+1}},A_{\overline{\mathcal{X}}}^{\mathcal{K}_{i+1}}|S_{\mathcal{X}},Q_{[N]}^{\mathcal{K}_{i+1}},W_{\mathcal{K}_1\cup\cdots\cup\mathcal{K}_{i}})\label{eq:thmmmxstpir26}\\
&=\frac{T}{N-X}\left(H_q(W_{\mathcal{K}_{i+1}}|S_{\mathcal{X}},Q_{[N]}^{\mathcal{K}_{i+1}},W_{\mathcal{K}_1\cup\cdots\cup\mathcal{K}_{i}})+H_q(A_{\overline{\mathcal{X}}}^{\mathcal{K}_{i+1}}|S_{\mathcal{X}},Q_{[N]}^{\mathcal{K}_{i+1}},W_{\mathcal{K}_1\cup\cdots\cup\mathcal{K}_{i+1}})\right)\label{eq:thmmmxstpir27}\\
&=\frac{T}{N-X}\left(ML+H_q(A_{\overline{\mathcal{X}}}^{\mathcal{K}_{i+1}}|S_{\mathcal{X}},Q_{[N]}^{\mathcal{K}_{i+1}},W_{\mathcal{K}_1\cup\cdots\cup\mathcal{K}_{i+1}})\right).\label{eq:thmmmxstpir28}
\end{align}
\eqref{eq:thmmmxstpir25} follows from the Han's inequality, and \eqref{eq:thmmmxstpir26} follows from \eqref{eq:mmxstpiransfunc} and \eqref{eq:mmxstpirdecode}. \eqref{eq:thmmmxstpir26} is the chain rule of entropy, while the last step holds from Lemma \ref{lemma:mmxstpir4} and the fact that $\mathcal{K}_i$'s are disjoint sets. Now let us apply \eqref{eq:thmmmxstpir28} repeatedly for $i=1,2,\dots,\lfloor\frac{K}{M}\rfloor-1$, we have
\begin{align}
ML&\leq\sum_{n\in\overline{\mathcal{X}}}D_n-H_q(A_{\overline{\mathcal{X}}}^{\mathcal{K}_1}|S_{\mathcal{X}},Q_{[N]}^{\mathcal{K}_1},W_{\mathcal{K}_1})\\
&\leq\sum_{n\in\overline{\mathcal{X}}}D_n-\frac{T}{N-X}\left(ML+H_q(A_{\overline{\mathcal{X}}}^{\mathcal{K}_2}|S_{\mathcal{X}},Q_{[N]}^{\mathcal{K}_2},W_{\mathcal{K}_1\cup\mathcal{K}_2})\right)\\
&\leq\dots\\
&\leq\sum_{n\in\overline{\mathcal{X}}}D_n-ML\left(\left(\frac{T}{N-X}\right)+\dots+\left(\frac{T}{N-X}\right)^{\lfloor\frac{K}{M}\rfloor-1}\right).
\end{align}
Thus we have
\begin{align}
\sum_{n\in\overline{\mathcal{X}}}D_n\geq ML\left(1+\left(\frac{T}{N-X}\right)+\dots+\left(\frac{T}{N-X}\right)^{\lfloor\frac{K}{M}\rfloor-1}\right).
\end{align}
Averaging over all $\mathcal{X}$, we have
\begin{align}
D=\sum_{n\in[N]}D_n\geq ML\frac{N}{N-X}\left(1+\left(\frac{T}{N-X}\right)+\dots+\left(\frac{T}{N-X}\right)^{\lfloor\frac{K}{M}\rfloor-1}\right).
\end{align}
Therefore,
\begin{align}
R=\frac{ML}{D}\leq \frac{N-X}{N}\left(1+\left(\frac{T}{N-X}\right)+\dots+\left(\frac{T}{N-X}\right)^{\lfloor\frac{K}{M}\rfloor-1}\right)^{-1}.
\end{align}
So we have,
\begin{align}
&C_{\text{MM-XSTPIR}}(N,X,T,K,M)\notag\\
&\leq \frac{N-X}{N}\left(1+\left(\frac{T}{N-X}\right)+\dots+\left(\frac{T}{N-X}\right)^{\lfloor\frac{K}{M}\rfloor-1}\right)^{-1},\quad N>X+T.
\end{align}
This completes the proof of Theorem \ref{thm:mmxstpir}.
\end{proof}

\begin{remark}Note that when $X=0$, i.e., the basic multi-message $T$-private information retrieval problem where storage is not secure, the proof of Theorem \ref{thm:mmxstpir} follows directly, and the resulting upper bound is obtained by setting $X=0$.
\end{remark}

\begin{remark}Note that when $T=0$, i.e., the problem with $X$-secure storage and no privacy requirement, we have
\begin{align}
&C_{\text{MM-XSTPIR}}(N,X,T=0,K,M)\notag\\
\leq&\left\{
\begin{aligned}
&0,&&N\leq X,\\
&\frac{N-X}{N},&&N>X.
\end{aligned}
\right.
\end{align}
\end{remark}

\section{Proof of Lemma \ref{lemma:Hab}}\label{proof:lemmaHab}
To prove Lemma \ref{lemma:Hab}, we need the following  lemmas. \begin{lemma}\label{lemma:invert}
For independent random matrices $\bar{\bf A}\in\mathbb{F}_q^{l\times k}, \bar{\bf B}\in\mathbb{F}_q^{k\times k}$, if the elements of $\bar{\bf B}$ are i.i.d. uniform then
\begin{align}
\lim_{q\rightarrow\infty}H_q({\bf \bar{A}\bar{B}}\mid \bar{\bf B})&= H_q(\bar{\bf A})
\end{align}
in $q$-ary units.
\end{lemma}
\begin{proof}
Define $\sigma$ as $0$ if $\bar{\bf B}$ is singular, and $1$ otherwise. Then we have
\begin{align}
H_q({\bf \bar{A}\bar{B}}\mid \bar{\bf B})&=H_q({\bf \bar{A}\bar{B}}\mid {\bf \bar{B}},\sigma)\\
&=H_q({\bf \bar{A}\bar{B}}\mid {\bf \bar{B}}, \sigma=1)P(\sigma=1)+H_q({\bf \bar{A}\bar{B}}\mid {\bf \bar{B}}, \sigma=0)P(\sigma=0)\\
&=H_q({\bf \bar{A}}\mid {\bf \bar{B}}, \sigma=1)P(\sigma=1)+H_q({\bf \bar{A}\bar{B}}\mid {\bf \bar{B}}, \sigma=0)P(\sigma=0)\label{eq:invertB}\\
&=H_q({\bf \bar{A}})P(\sigma=1)+H_q({\bf \bar{A}\bar{B}}\mid {\bf \bar{B}}, \sigma=0)P(\sigma=0)\\
&=H_q({\bf \bar{A}})\prod_{i=1}^k(1-q^{-i})+H_q({\bf \bar{A}\bar{B}}\mid {\bf \bar{B}}, \sigma=0)\left(1-\prod_{i=1}^k(1-q^{-i})\right)\label{eq:pzero}
\end{align}
In \eqref{eq:invertB} we used the fact that given a square non-singular (invertible) matrix $\bar{\bf B}$, the matrix $\bar{\bf A}\bar{\bf B}$ is an invertible function of the matrix $\bar{\bf A}$. In \eqref{eq:pzero} we used the result from \cite{Waterhouse} that the probability of a matrix $\bar{\bf B}$ drawn uniformly from $\mathbb{F}_q^{k\times k}$ being singular is exactly $1-\prod_{i=1}^k(1-q^{-1})$. Now, since $H_q({\bf \bar{A}\bar{B}}\mid {\bf \bar{B}}, \sigma=0)$ is a finite value bounded between $0$ and $lk$, as $q\rightarrow\infty$ we have $H_q({\bf \bar{A}\bar{B}}\mid \bar{\bf B})= H_q(\bar{\bf A})$.
\end{proof}
The random matrices  ${\bf A}, {\bf B}$ in the next two lemmas are as defined in Lemma \ref{lemma:Hab}. Note that we assume that $q\rightarrow\infty$ throughout the remainder of this section.

\begin{lemma}\label{lemma:hablm1}
When $K\geq M$, let us express $\mathbf{A}$ as
\begin{align}
\mathbf{A}=[\underbrace{(\mathbf{A}_1)_{L\times M}}_{\text{First $M$ columns}}|\underbrace{(\mathbf{A}_2)_{L\times (K-M)}}_{\text{Last $K-M$ columns}}].
\end{align}
Similarly, let us express $\mathbf{B}$ as
\begin{align}
\mathbf{B}=\begin{bmatrix}
(\mathbf{B}_1)_{M\times M}\\
(\mathbf{B}_2)_{(K-M)\times M}
\end{bmatrix}
\begin{array}{l}
\}\text{\scriptsize First $M$ rows}\\
\}\text{\scriptsize Last $K-M$ rows.}
\end{array}
\end{align}
Then we have
\begin{align}
H_q(\mathbf{A}\mathbf{B}\mid\mathbf{B}_1,\mathbf{B}_2,\mathbf{A}_2)=H_q(\mathbf{A}_1).
\end{align}
\end{lemma}
\begin{proof}
As $q\rightarrow\infty$, the square matrix $\mathbf{B}_1$ is invertible with probability $1$. Therefore, using Lemma \ref{lemma:invert}, {$H_q({\bf AB}\mid {\bf B}_1,{\bf B}_2,{\bf A}_2)=H_q({\bf A}_1{\bf B}_1+{\bf A}_2{\bf B}_2\mid {\bf B}_1,{\bf B}_2,{\bf A}_2)=H_q({\bf A}_1{\bf B}_1\mid {\bf B}_1,{\bf B}_2,{\bf A}_2)=H_q({\bf A_1}\mid {\bf B}_1,{\bf B}_2,{\bf A}_2)=H_q({\bf A}_1)$}.

\end{proof}

\begin{lemma}\label{lemma:hablm2}
When $K<M$, let us express $\mathbf{B}$ as
\begin{align}
\mathbf{B}=[\underbrace{(\mathbf{B}_1)_{K\times K}}_{\text{First $K$ columns}}|\underbrace{(\mathbf{B}_2)_{K\times (M-K)}}_{\text{Last $M-K$ columns}}].
\end{align}
Then we have
\begin{align}\label{eq:lemmahab21}
H_q(\mathbf{A}\mathbf{B}\mid \mathbf{B}_1,\mathbf{B}_2)=H_q(\mathbf{A}).
\end{align}
In particular, when $K<L$, we have
\begin{align}\label{eq:lemmahab22}
H_q(\mathbf{A}\mathbf{B}\mid\mathbf{B}_1)=H_q(\mathbf{A})+H_q(\mathbf{B}_2).
\end{align}
\end{lemma}
\begin{proof}
As $q\rightarrow\infty$, the square matrix $\mathbf{B}_1$ is invertible with probability $1$. Therefore, $H_q({\bf AB}\mid {\bf B}_1,{\bf B}_2)=H_q({\bf A}{\bf B}_1, {\bf A}{\bf B}_2\mid {\bf B}_1,{\bf B}_2)=H_q({\bf A}, {\bf A}{\bf B}_2\mid {\bf B}_1,{\bf B}_2)=H_q({\bf A}\mid {\bf B}_1,{\bf B}_2)=H_q({\bf A}).$
When $K<L$, the matrix $\mathbf{A}$ has full column rank with probability $1$, so that given ${\bf A}$, the matrix ${\bf B}_2$ is an invertible function of ${\bf A}{\bf B}_2$. Therefore, $H_q({\bf AB}\mid {\bf B}_1)=H_q({\bf AB_1, AB_2}\mid {\bf B}_1)=H_q({\bf A, AB_2}\mid {\bf B}_1)=H_q({\bf A}\mid {\bf B}_1)+H_q({\bf AB_2}\mid {\bf B}_1, {\bf A})=H_q({\bf A}\mid {\bf B}_1)+H_q({\bf B}_2\mid {\bf B}_1,{\bf A})=H_q({\bf A})+H_q({\bf B}_2).$

\end{proof}
Now we are ready to prove the first part of Lemma \ref{lemma:Hab}.
\subsection{Proof of Lemma \ref{lemma:Hab}: \eqref{eq:Hab}}
\begin{proof}

Case 1. $K\leq\min(L,M)$.

First, let us consider the upper bound. Note that we can rewrite matrix $\mathbf{A}$ as
\begin{align}
\mathbf{A}=\begin{bmatrix}
(\mathbf{A}_1)_{K\times K}\\
(\mathbf{A}_2)_{(L-K)\times K}\\
\end{bmatrix}\begin{array}{l}
\}\text{\scriptsize First $K$ rows}\\
\}\text{\scriptsize Last $L-K$ rows}.
\end{array}
\end{align}
Similarly, let us rewrite matrix $\mathbf{B}$ as
\begin{align}
\mathbf{B}=[\underbrace{(\mathbf{B}_1)_{K\times K}}_{\text{First $K$ columns}}|\underbrace{(\mathbf{B}_2)_{K\times (M-K)}}_{\text{Last $M-K$ columns}}].
\end{align}
Note that as $q\rightarrow\infty$, square matrices $\mathbf{A}_1$ and $\mathbf{B}_1$ are invertible with probability $1$. Thus we have
\begin{align}
\mathbf{A}=\begin{bmatrix}
\mathbf{I}_K\\
\mathbf{A}_2\mathbf{A}_1^{-1}\\
\end{bmatrix}\mathbf{A}_1,\\
\mathbf{B}=\mathbf{B}_1~\left[\mathbf{I}_K\mid\mathbf{B}_1^{-1}\mathbf{B}_2\right].
\end{align}
Therefore, 
\begin{align}
\mathbf{A}\mathbf{B}=\begin{bmatrix}
\mathbf{I}_K\\
\mathbf{A}_2\mathbf{A}_1^{-1}\\
\end{bmatrix}(\mathbf{A}_1\mathbf{B}_1)\left[\mathbf{I}_K|\mathbf{B}_1^{-1}\mathbf{B}_2\right].
\end{align}
Then we have
\begin{align}
H_q(\mathbf{A}\mathbf{B})&\leq H_q(\mathbf{A}_2\mathbf{A}_1^{-1},\mathbf{A}_1\mathbf{B}_1,\mathbf{B}_1^{-1}\mathbf{B}_2)\\
&\leq K(L-K)+K^2+K(M-K)\\
&=LK+KM-K^2,
\end{align}
in $q$-ary units. On the other hand, from \eqref{eq:lemmahab22} of Lemma \ref{lemma:hablm2}, we have
\begin{align}
H_q(\mathbf{A}\mathbf{B})\geq &H_q(\mathbf{A}\mathbf{B}\mid\mathbf{B}_1)\\
=&H_q(\mathbf{A})+H_q(\mathbf{B}_2)\\
=&LK+KM-K^2,
\end{align}
in $q$-ary units. This completes the proof of \eqref{eq:Hab} for Case 1.

Case 2. $M\leq\min(L,K)$.

Let us consider the upper bound first. Since $\mathbf{A}\mathbf{B}$ has dimension $L\times M$,  it is trivial that $H_q(\mathbf{A}\mathbf{B})\leq LM$ in $q$-ary units. On the other hand, from Lemma \ref{lemma:hablm1}, we have $H_q(\mathbf{A}\mathbf{B})\geq H_q(\mathbf{A}\mathbf{B}\mid\mathbf{B}_1,\mathbf{B}_2,\mathbf{A}_2)=H_q(\mathbf{A}_1)=LM$ 
in $q$-ary units. This completes the proof of \eqref{eq:Hab} for Case 2.

Case 3. $L\leq\min(K,M)$. By symmetry this case is identical to Case 2. This completes the  proof of \eqref{eq:Hab} for Lemma \ref{lemma:Hab}.
\end{proof}
\subsection{Proof of Lemma \ref{lemma:Hab}: \eqref{eq:Haba},\eqref{eq:Habb}}

First let us prove that $H_q(\mathbf{A} \mathbf{B} \mid \mathbf{B})=\min(LM,LK)$. If $K\leq M$, then from \eqref{eq:lemmahab21} of Lemma \ref{lemma:hablm2}, we have $H_q(\mathbf{A}\times \mathbf{B}\mid\mathbf{B})=H_q(\mathbf{A})=LK$ in $q$-ary units. Now consider $K>M$.  We have $H_q(\mathbf{A} \mathbf{B}\mid \mathbf{B})\leq H_q(\mathbf{A} \mathbf{B})= LM$, in $q$-ary units. On the other hand, from Lemma \ref{lemma:hablm1}, we have $H_q(\mathbf{A}\times \mathbf{B}\mid \mathbf{B})=H_q(\mathbf{A}\times \mathbf{B}\mid \mathbf{B}_1,\mathbf{B}_2)
\geq H_q(\mathbf{A}\times \mathbf{B}\mid \mathbf{B}_1,\mathbf{B}_2,\mathbf{A}_2)=H_q(\mathbf{A}_1)=LM$ in $q$-ary units. By symmetry, $H_q(\mathbf{A}\times \mathbf{B} | \mathbf{A})=\min(LM,KM)$ can be similarly proved. This completes the proof of Lemma \ref{lemma:Hab}.

\section{Invertibility of $\mathbf{M}_N$}\label{app:mn}
\begin{lemma}\label{lemma:mninv}
The matrix $\mathbf{M}_N$ is invertible if $f_{[S]}$ are $S$ distinct elements from $\mathbb{F}_q$ and $\alpha_{[N]}$ are $N$ distinct elements from $\mathbb{G}$,
\begin{equation}
\mathbb{G}=\{\alpha\in\mathbb{F}_q:\alpha+f_s\neq0, \forall s\in[S]\}.
\end{equation}
\end{lemma}

\begin{proof}
Proof of Lemma \ref{lemma:mninv} is almost identical to that of Lemma 5 in \cite{Jia_Sun_Jafar_XSTPIR}. The only difference is that since $f_{[S]}$ are distinct $S$ elements from $\mathbb{F}_q$, (117) in \cite{Jia_Sun_Jafar_XSTPIR} becomes
\begin{align}
g(\alpha)=\sum_{i=1}^{S}c_i\left(\frac{\Delta}{f_i+\alpha}\right).
\end{align}
Now choosing $\alpha$ such that $(f_i+\alpha)=0$ gives us $c_1=\dots=c_S=0$. Other parts of the proof in \cite{Jia_Sun_Jafar_XSTPIR} apply directly. This completes the proof of  Lemma \ref{lemma:mninv}.
\end{proof}

\bibliographystyle{IEEEtran}
\bibliography{Thesis}

\end{document}